\documentclass[10pt,twocolumn]{article}

\usepackage[
  letterpaper,
  margin=0.75in,
  columnsep=0.25in
]{geometry}

\usepackage{libertine}
\usepackage[libertine,amsthm]{newtxmath}
\usepackage[T1]{fontenc}

\usepackage[T1]{fontenc}    % use 8-bit T1 fonts
\usepackage{hyperref}       % hyperlinks
\usepackage{url}            % simple URL typesetting
\usepackage{booktabs}       % professional-quality tables
\usepackage{amsfonts}       % blackboard math symbols
\usepackage{nicefrac}       % compact symbols for 1/2, etc.
\usepackage{microtype}      % microtypography
\usepackage{lipsum}		% Can be removed after putting your text content
\usepackage{graphicx}
\usepackage{natbib}
\usepackage{doi}

\usepackage{amsmath}
\usepackage{enumitem}
\usepackage{caption}
\usepackage{subcaption}
\usepackage[]{graphicx}
\usepackage[ruled,linesnumbered,noend]{algorithm2e} 
\usepackage{mathtools}
\usepackage{bbm}
\usepackage{amsthm}
\usepackage{multirow}
\usepackage{dsfont}
\usepackage{thm-restate}

\usepackage{tcolorbox}
\tcbuselibrary{breakable, skins}

\newtcolorbox{topbox}[2][]{%
  enhanced,
  breakable,
  colback=white,
  colframe=black,
  fonttitle=\bfseries,
  boxed title style={colback=white, colframe=black},
  #1
}

\newcommand{\sone}{(\mathsf{S1})}
\newcommand{\stwo}{(\mathsf{S2})}
\newcommand{\sthree}{(\mathsf{S3})}
\newcommand{\sfour}{(\mathsf{S4})}
\newcommand{\mgx}{\mathsf{MG}(X, k)}
\newcommand{\mgxp}{\mathsf{MG}(X^{\prime}, k)}
\newcommand{\ssx}{\mathsf{SS}}
\newcommand{\pr}{^{\prime}}
\newcommand{\mg}{\mathsf{MisraGries}}
\newcommand{\spacesaving}{\mathsf{SpaceSaving}}
\newcommand{\cs}{\mathsf{CountSketch}}
\newcommand{\cms}{\mathsf{CountMinSketch}}

\newcommand{\laplacedist}{\mathtt{Laplace}}

\newcommand{\prob}{\mathds{P}}

\newcommand{\diffpriv}{\mathsf{DP}}

\newtheorem{theorem}{Theorem}
\newtheorem{lemma}{Lemma}

\newtheorem{definition}{Definition}
\newtheorem{corollary}{Corollary}

\usepackage{microtype}
\usepackage{graphicx}
\usepackage{hyperref}

\setlist{nosep}

\usepackage{titlesec}
\titleformat{\section}
  {\large\bfseries}
  {\thesection}{0.5em}{}
\titleformat{\subsection}
  {\normalsize\bfseries}
  {\thesubsection}{0.5em}{}

\title{An Iconic Heavy Hitter Algorithm Made Private}
\author{Rayne Holland\footnote{raynehollandphd@gmail.com}}
\date{}

\begin{document}
\maketitle

\begin{abstract}
    Identifying heavy hitters in data streams is a fundamental problem with widespread applications in modern analytics systems.
    These streams are often derived from sensitive user activity, making adequate privacy guarantees necessary.
    While recent work has adapted the classical heavy hitter algorithm Misra-Gries to satisfy differential privacy in the streaming model, the privatization of other heavy hitter algorithms with \emph{better} empirical utility is absent.  

    Under this observation, we present the first differentially private variant of the SpaceSaving algorithm, which, in the non-private setting, is regarded as the state-of-the-art in practice.
    Our construction post-processes a non-private SpaceSaving summary by injecting asymptotically optimal noise and applying a carefully calibrated selection rule that suppresses unstable labels.
    This yields strong privacy guarantees while preserving the empirical advantages of SpaceSaving.

    Second, we introduce a generic method for extracting heavy hitters from any differentially private frequency oracle in the data stream model.
    The method requires only $\mathcal{O}(k)$ additional memory, where $k$ is the number of heavy items, and provides a mechanism for safely releasing item identities from noisy frequency estimates.
    This yields an efficient, plug-and-play approach for private heavy hitter recovery from linear sketches.

    Finally, we conduct an experimental evaluation on synthetic and real-world datasets.
    Across a wide range of privacy parameters and space budgets, our method provides superior utility to the existing differentially private Misra–Gries algorithm.
    Our results demonstrate that the empirical superiority of SpaceSaving survives privatization and that efficient, practical heavy hitter identification is achievable under strong differential privacy guarantees.
\end{abstract}

% ---------- Main content ----------

\section{Introduction}
\label{sec:intro}

Identifying the most frequent elements of a dataset (commonly referred to as ``heavy hitters'') is a cornerstone problem in data analysis~\cite{cormode2008finding}.
In modern systems, this task commonly arises on massive data streams, where items arrive continuously and must be processed online under strict memory and time constraints.
Prominent solutions include counter-based methods~\cite{misra1982finding,mitzenmacher2012hierarchical} and linear sketches~\cite{charikar2002finding,cormode2005improved}.
However, data streams are often derived from sensitive user activity~\cite{wang2021continuous}, and the classic techniques do not provide adequate privacy.

Recent work seeks to retrofit these algorithms with \textit{differential privacy} ($\diffpriv$) guarantees, aiming for minimal overhead and limited impact on utility~\cite{biswas2024differentially,lebeda2023better,pagh2022improved,zhao2022differentially}.
At a high level, $\diffpriv$ ensures that the outcome of an analysis does not change significantly under small changes to the underlying dataset, thereby limiting what can be inferred about any single update.
This line of inquiry commonly adopts the \emph{single-observation} streaming model: the stream is processed in one pass and a single $\diffpriv$ solution is released.
The setting reflects common analytics pipelines that aggregate large volumes of data and publish sanitized summaries under a fixed privacy budget.

From an algorithmic perspective, enforcing $\diffpriv$ for heavy hitters exposes two distinct challenges.
First, sufficient noise must be added to protect individual updates without overwhelming the frequency signal of genuinely heavy items.
Second, the algorithm must report item {identities}, even though the set of reported items may change between similar streams.
Consequently, identity disclosure is a source of privacy leakage.
Prior work has addressed these challenges along two complementary lines.
One line focuses on the counter-based method \textit{Misra-Gries} ($\mg$), which explicitly maintains candidate frequent items and their counts \cite{biswas2024differentially, chan2012differentially, lebeda2023better}.
The other centers on  utilizing $\diffpriv$ linear sketches as frequency oracles within a (labor intensive) heavy hitter search algorithm \cite{bassily2017practical}.

Building on initial work by Chen et al. \cite{chan2012differentially}, Lebeda and Tetek adapted the $\mg$ algorithm to enable $\diffpriv$ releases with asymptotically optimal noise error~\cite{lebeda2023better}.
This result is notable as the authors observe that ``in practice, non-private approximate histograms are often computed using the Misra-Gries sketch''.
Nevertheless, empirical studies consistently show that the counter-based SpaceSaving algorithm ($\spacesaving$) \cite{mitzenmacher2012hierarchical} substantially outperforms $\mg$ (and linear sketches) in utility \cite{cormode2008finding}.
This gap raises the natural question as to whether $\spacesaving$ can be adapted to satisfy $\diffpriv$ in the single-observation streaming model.
In response, we provide the first method to release $\spacesaving$ under $\diffpriv$ and demonstrate that it retains its empirical dominance in practice. 

In parallel, {linear sketches} have been adapted to admit $\diffpriv$ {frequency estimates} with strong utility~\cite{pagh2022improved,zhao2022differentially}.
However, as idenity disclosure is a form of privacy leakage, uncritically converting these $\diffpriv$ estimates into released item {identities} is not, in itself, differentially private.
%The challenge is reporting frequent items without leaking the presence or absence of isolated labels across similar streams.
A naive brute-force solution is to query the entire input domain, a prohibitive computational effort in most settings.
To address this, Bassily et al. propose a more efficient heavy-hitter search procedure that utilizes a linear sketch as a $\diffpriv$ frequency oracle in post-processing \cite{bassily2017practical}.
Despite this improvement, the approach still incurs substantial computational and memory overheads that exceed the budgets typically assumed in the data stream model.
Building on insights from our $\diffpriv$ release of $\spacesaving$, we introduce a generic method for extracting heavy hitters from any $\diffpriv$ frequency oracle using only an additional $\mathcal{O}(k)$ memory, where $k$ is the number of heavy items.
In contrast to prior work, this $\mathcal{O}(k)$ overhead is compatible with standard streaming constraints.

\subsection{Our contributions}
In providing novel methods for privatizing existing algorithms for identifying heavy hitters, our contributions can be summarized as follows.
\begin{itemize}
\item 
We present the first $\diffpriv$ variant of the $\spacesaving$ algorithm in the single-observation streaming model, achieving asymptotically optimal noise while preserving the empirical performance advantages of non-private $\spacesaving$.
\item
We introduce a general framework for releasing heavy hitters from any $\diffpriv$ frequency oracle in the single-observation model, with an additional memory overhead of $\mathcal{O}(k)$, where $k$ is the maximum number of heavy items. 
%The additional memory cost is proportional to the number of heavy items and can be immediately utilized with existing $\diffpriv$ linear sketches.
\item 
We conduct an experimental evaluation on synthetic and real-world datasets, comparing our methods against state-of-the-art $\diffpriv$ $\mg$. 
Across a wide range of privacy parameters and space budgets, our $\diffpriv$ $\spacesaving$ consistently achieves higher utility while matching $\diffpriv$ $\mg$ in memory usage and throughput.
\end{itemize}

\subsection{Algorithm overview}

Counter-based streaming algorithms explicitly maintain sets of candidate frequent item \emph{identities}.
As these sets may differ between similar streams, the disclosure of item identities itself becomes a source of privacy leakage.
Thus, a central challenge in privatizing a counter-based algorithm is designing a post-processing mechanism that releases frequent item sets without exposing the privacy leakage.

Neighboring streams $X, X\pr$ are defined as streams that differ in one update ($X = X\pr \cup \{x\}$).
At the heart of our approach is a strategy of characterizing how tracked sets change across {neighboring} streams.
This allows us to distinguish \emph{stable} labels from \emph{unstable} ones and to release only the former in a $\diffpriv$-safe manner.
Our characterization states that unstable labels always observe a frequency estimate close to the minimum estimate across tracked items.
Thus, unstable labels can be suppressed with a sufficiently high threshold.

Operationally, our $\diffpriv$ $\spacesaving$  runs a non-private $\spacesaving$ algorithm with a slightly expanded capacity. 
At release time, we add independent Laplace noise to the counters and apply a single threshold that both (i) suppresses labels that can differ across neighboring inputs and (ii) retains true heavy hitters with high probability.
The expanded counter capacity allows us to reduce the threshold and, therefore, target unstable labels without eliminating true stable heavy hitters.

In addition, we extend these ideas to any frequency oracle with provable error guarantees.
We show that the same threshold selection principle can be adapted to any oracle with provable error guarantees (what we term the error envelope).
This enables a plug-and-play solution for $\diffpriv$ linear sketches.

\subsection{Paper Outline}

Section~\ref{sec:prelims} provides the relevant preliminaries and problem definition.
Section~\ref{sec:priors} details counter-based algorithms and linear sketches (an expanded related work is available in Appendix~\ref{app:related_work}).
Section~\ref{sec:soss} introduces the $\diffpriv$ $\spacesaving$ algorithm.
Section~\ref{sec:lin_sketch} provides a more generic method for $\diffpriv$ heavy hitter output that operates on any $\diffpriv$ frequency oracle.
Section~\ref{sec:experiments} provides an empirical evaluation.

\section{Preliminaries}
\label{sec:prelims}

\subsection{Privacy}

$\diffpriv$ ensures that the inclusion or exclusion of any individual in a dataset has a minimal and bounded impact on the output of a mechanism.
Two streams ${X} = \{x_1, \ldots, x_T\}$ and ${X}^{\prime}= \{x_1^{\prime}, \ldots, x_T^{\prime}\}$ are \textit{neighboring}, 
denoted ${X} \sim {X}^{\prime}$,
if they differ by the addition or removal of a single update.
Formally, ${X} \sim {X}^{\prime}$ if $X = X\pr \cup \{x\}$.
The following definition of $\diffpriv$ is adapted from Dwork and Roth~\cite{dwork2014algorithmic}.
\begin{definition}[\((\varepsilon, \delta)\)-Differential Privacy]
Let ${X}, {X}^{\prime} \in \mathcal{U}^*$ be neighboring input streams (i.e., differing in at most one element). A randomized mechanism $\mathcal{M}$ satisfies \((\varepsilon, \delta)\)-differential privacy under a given observation model if for all measurable sets of outputs \( Z \subseteq \mathtt{support}(\mathcal{M}) \), it holds that:
\[
\prob[\mathcal{M}({X}) \in Z] \leq e^{\varepsilon} \cdot \prob[\mathcal{M}({X}^{\prime}) \in Z] + \delta.
\]
\label{def:diff_privacy}
\end{definition}
\noindent
In the single-observation model, the mechanism processes the stream once and releases a single $\diffpriv$ output after the final stream update has been observed.
%In this model, neighboring streams differ. 

A common method for outputting $\diffpriv$ functions is the following
\begin{lemma}[Laplace Mechanism for Vector Valued Functions]
\label{lem:laplace_vector}
Let $f: \mathcal{X} \to \mathbb{R}^d$ be a function with $\ell_1$-sensitivity at most $s$, that is, for any pair of neighboring inputs $X, X' \in \mathcal{X}$,
\[
\|f(X) - f(X')\|_1 \le s.
\]
Let $\varepsilon > 0$ and define the randomized mechanism
\[
\mathcal{M}(X) := f(X) + (Z_1, \dots, Z_d),
\]
where each $Z_i \sim \laplacedist(s/\varepsilon)$ independently. Then $\mathcal{M}$ satisfies $(\varepsilon, 0)$-$\diffpriv$.
\end{lemma}

\subsection{Heavy Hitter Detection}

We consider a data-stream model in which a sequence of items \( X = \{x_1,x_2,\dots,x_t \}\) from a universe \(\mathcal U\) arrive in one pass. 
For each item \(y\in\mathcal U\), its true frequency is
\[
f_y = \sum_{x_j \in X} \mathrm{1}\{x_j = y\},
\]
The frequency vector is \(\mathbf f = (f_y)_{y\in\mathcal U}\).
Two canonical problems are:
\begin{enumerate}
    \item \emph{Frequency estimation:} for any query item \(y\), return \(\widehat f_y\) satisfying \(\widehat f_y \le f_y \le \widehat f_y + \eta t\) for some parameter \(\eta>0\).
    \item \emph{Heavy-hitter detection:} given thresholds \(k\) and \(\eta\), return a set \(H\subseteq \mathcal U\) such that every item $y$ with \(f_y \ge  n/k\) belongs to \(H\) and no item $z$ with \(f_z \le (1/k - \eta)n\) is included.
\end{enumerate}

\section{Prior Methods}
\label{sec:priors}
We now introduce background on counter-based methods and linear sketches.
A broader discussion of prior work is provided in Appendix~\ref{app:related_work}.
\subsection{Counter Methods}

\begin{algorithm}[t]
  \SetAlgoLined
  \DontPrintSemicolon
  \SetKwProg{myproc}{define}{}{}
    \myproc{$\mathsf{MG}(X, k)$}{
        $\mathcal{T}_0, C_0 \gets \varnothing$\;
        \For{$x_t \in X$}{
            \If{$x_t \in \mathcal{T}_{t-1}$}{$C_t[x_t] \gets C_{t-1}[x_t] +1$}
            \ElseIf{$|\mathcal{T}_{t-1}| < k$}{
                $\mathcal{T}_t\gets \mathcal{T}_{t-1} \cup \{x_t\}$\;
                $C_t[x_t] \gets 1$\;
            }
            \Else{
                \For{$y \in \mathcal{T}_{t-1}$}{
                    $C_t[y] \gets C_{t-1}[y] -1$\;
                }
                $S \gets \{ y \mid C_t[y] = 0\}$\;
                $\mathcal{T}_t \gets \mathcal{T}_{t-1} \setminus S$\;
            }
        }
        \KwRet $\mathcal{T}_T, C_T$
    }
  \caption{$\mg$}
  \label{alg:MG}
\end{algorithm}

\begin{algorithm}[t]
  \SetAlgoLined
  \DontPrintSemicolon
  \SetKwProg{myproc}{define}{}{}
  \SetKwInOut{require}{let}
  \require{$\tau_t[a] = \max \{\, t' \mid x_{t'} = a,\; t' \le t \,\}$ }
    \myproc{$\mathsf{SS}(X, k)$}{
        $\mathcal{T}_0, C_0 \gets \varnothing$\;
        \For{$x_t \in X$}{
            \If{$x_t \in \mathcal{T}$}{$C_t[x_t] \gets C_{t-1}[x_t] +1$}
            \ElseIf{$|\mathcal{T}_{t-1}| < k$}{
                $\mathcal{T}_t\gets \mathcal{T}_{t-1} \cup \{x_t\}$\;
                $C_t[x_t] \gets 1$
            }
            \Else{
                $S \gets \arg\min_{a \in \mathcal{T}_{t-1}} C_{t-1}[a]$\;
                $y \gets \arg\max_{a \in S} \tau_{t-1}[a]$\;
                $C_t[x_t] \gets C_{t-1}[y] + 1; C_t[y] \gets 0 $\;
                $\mathcal{T}_t \gets (\mathcal{T}_{t-1} \setminus \{y\}) \cup \{x_t\}$
            }
        }
        \KwRet $\mathcal{T}_T, C_T$
    }
  \caption{$\spacesaving$}
  \label{alg:SS}
\end{algorithm}
Counter-based methods maintain a small table of  \(\mathcal{O}(k)\)  (item,counter) pairs. 
On each stream update, these methods adjust  tracked items and corresponding counters so that one obtains error guarantees on both frequency estimation and heavy-hitter recovery. 
These methods are attractive for their simplicity and deterministic error bounds.

\subsubsection{Misra Gries}

The $\mg$ algorithm maintains a set $\mathcal{T}$ and corresponding counters $C$, both of size at most $k-1$. 
For each incoming item, if it is already present in the table $\mathcal{T}$, its counter is incremented. 
If it is not present and a counter has value zero, that counter is assigned to the new item with count set to 1. 
If all $k-1$ counters are occupied by distinct items, all counters are decremented by 1. 
For a stream of length $T$, a simple grouping argument shows that any item appearing more than $T/k$ times in the stream is guaranteed to be stored in the table when the algorithm terminates.
The process is formalized in Algorithm~\ref{alg:MG}.

To privatize a counter-based algorithm it is not sufficient to inject noise into the counts of the tracked items.
This is because the identities of the tracked items can differ on neighboring inputs, allowing an adversary to easily distinguish the streams based on the output of the algorithm.
This notion is formalized in the following result by Lebeda and Tetek \cite{lebeda2023better}, which describes the observed outcomes when a $\mg$ algorithm processes neighboring input streams.

\begin{lemma}[\cite{lebeda2023better}]
    Let $X = X^{\prime} \cup \{x\}$.
    Let $(\mathcal{T}, {C})  \gets \mathtt{MG}(X, k)$ and $(\mathcal{T}^{\prime}, {C}^{\prime}) \gets \mathtt{MG}(X^{\prime}, k)$ denote the outputs of the Misra Gries algorithm on inputs $X$ and $X^{\prime}$.
    Then $|\mathcal{T} \cap \mathcal{T}^{\prime}| \geq k-2$; for all $y\notin \mathcal{T} \cap \mathcal{T}^{\prime}$, $C[y]\leq 1$ and $C^{\prime}[y] \leq 1$; and exactly one of the following is true:
    \begin{enumerate}
        \item[(1)] $\exists i \in \mathcal{T}$, such that $C[i] = C^{\prime}[i]+1$, and $\forall j \neq i: C[j] = C^{\prime}[j]$.
        \item[(2)] $\forall i \in \mathcal{T}^{\prime}: C[i] = C^{\prime}[i]-1$, and $C^{\prime}[j] = 0 $ for $j \notin \mathcal{T}^{\prime}$
    \end{enumerate}
    \label{lem:1p_MG}
\end{lemma}
The lemma states that there can be at most 2 elements that are in $\mathcal{T}$, but not in $\mathcal{T}^{\prime}$, and vice versa.
These elements are referred to as \textit{isolated elements}, and if they appear in the final output, an adversary can always tell if an output was generated by $X$ or $X^{\prime}$.
Furthermore, the count of an isolated element must always be at most 1. 
Therefore, we can use a thresholding trick to suppress isolated elements from the output with high probability.
Once isolated elements are suppressed, and appropriate noise is added to item counts, tracked items with noisy counts above the threshold can be released.
%Their solution has error $t/(k+1) + \mathcal{O}(\log (1/\delta)/\varepsilon)$, which is asymptotically optimal.

\subsubsection{Space Saving}

The $\spacesaving$ algorithm has the same component data structures ($\mathcal{T}, C$) as $\mg$. 
When an item $x$ arrives, if it is already present in the table, its counter is incremented. 
If it is not present and a free counter exists, the item is inserted with count set to 1. 
In contrast to $\mg$,
if all $k$ counters are occupied, the item with the minimum count $c_{\min}$ is replaced by $i$, and its counter is set to $C[x] = c_{\min} + 1$. 
When multiple items observe the minimum count ties can be broken arbitrarily.
However, to enable privatization, we enforce the rule that the item that \textit{most recently} occurred is evicted. 
The process is formalized in Algorithm~\ref{alg:SS}.
This procedure guarantees that for every item, the estimated frequency satisfies the following bound.

\begin{lemma}[Additive Error of the $\spacesaving$]
    Let the $\spacesaving$ algorithm process a stream of length $T$ while maintaining  $k$ counters. 
    For each item $i$ with true frequency $f_i$ and estimated frequency 
    $\widehat f_i$, the following holds:
    \[
      f_i \le \widehat f_i \le f_i + \frac{T}{k}.
    \]
    \label{lem:ss_bound}
\end{lemma}

While $\spacesaving$ has the same absolute error bound as $\mg$, empirical studies show that it has superior utility in practice \cite{cormode2008finding}.
Currently, no $\diffpriv$ version of $\spacesaving$ exists.
This is the main gap addressed by our paper.

\subsection{Linear Sketches}

Linear sketches enable approximate frequency queries, in small memory, by hashing items to a compact collection of shared counters. 
We focus on the well known Count-min Sketch ($\cms$)~\cite{cormode2005improved} and Count Sketch ($\cs$)~\cite{charikar2002finding} variants. 
While they share a similar structure, they differ in their update and query procedures and offer different error trade-offs.

\subsubsection{Count-Min Sketch}
The $\cms$ uses $d$ hash functions $h_1, \ldots, h_d : \mathcal{U} \rightarrow [w]$ to maintain a $d \times w$ matrix of counters $C \in (\mathbb{R}^+)^{d \times w}$. 
For each incoming item $x \in \mathcal{U}$ the update rule is:
\[
    C[i, h_i(x)] \leftarrow C[i, h_i(x)] + 1, \quad \text{for all } i \in [d].
\]
%This update procedure is illustrated in Figure~\ref{fig:cms}.
To estimate the frequency of $x$, $\cms$ returns $\hat{f}_x = \min_{i \in [d]} C[i, h_i(x)]$.
This estimator is biased upwards due to hash collisions with other items but offers strong probabilistic error guarantees.

\subsubsection{Count Sketch}
The $\cs$  modifies the $\cms$ by incorporating randomized signs to reduce bias. 
In addition to the hash functions $h_1, \ldots, h_d : \mathcal{U} \rightarrow [w]$, it uses a second family of hash functions $g_1, \ldots, g_d : \mathcal{U} \rightarrow \{-1, +1\}$ that assign random signs.
For each $x \in \mathcal{U}$, the update rule is:
\[
    C[i, h_i(x)] \leftarrow C[i, h_i(x)] + g_i(x), \quad \text{for all } i \in [d].
\]
To estimate the frequency of $x$, $\cs$ returns:
\[
    \hat{f}_x = \mathrm{median}_{i \in [d]} \left( g_i(x) \cdot C[i, h_i(x)] \right).
\]
This estimator is unbiased and has error bounded in terms of the  second moment of the frequency vector.

\subsubsection{$\diffpriv$ Heavy Hitters with $\diffpriv$ Linear Sketches}

\begin{table*}[t]
\begin{topbox}[label={box:cases}]{Neighbouring Stream Cases}
\begin{itemize}
  \item[$\sone$] 
  \begin{itemize}
    \item[(a)] $\mathcal{T}_t = \mathcal{T}_t'$;
    \item[(b)] $\exists y \in \mathcal{T}_t$ such that $C_t[y] = C_t'[y]+1$;
    \item[(c)] $\forall z \in \mathcal{T}_t \setminus \{y\}, \; C_t[z] = C_t'[z]$.
  \end{itemize}
  \medskip\hrule\medskip
  \item[$\stwo$] 
  \begin{itemize}
    \item[(a)] $\mathcal{T}_t \neq \mathcal{T}_t'$; $W_t = \{z\}, \; W_t' = \{z'\}$;
    \item[(b)] $\forall a \in \mathcal{T}_t \cap \mathcal{T}_t', \; C_t[a] = C_t'[a]$;
    \item[(c)] $C_t[z] = C_t'[z']+1$; 
    $C_t'[z'] = \min_{a \in \mathcal{T}'} C_t'[a]$, 
    $C_t[z] \leq \min_{a \in \mathcal{T}} C_t[a] + 1$.
  \end{itemize} 
  \medskip\hrule\medskip
  \item[$\sthree$]
  \begin{itemize}
    \item[(a)] $\mathcal{T}_t \neq \mathcal{T}_t'$; $W_t = \{z\}, \; W_t' = \{z'\}$;
    \item[(b)] $\exists y \in \mathcal{T}_t \cap \mathcal{T}_t' $ such that $C_t[y]=C_t'[y]+1$, and 
    $\forall a \in \mathcal{T}_t \cap \mathcal{T}_t' \setminus \{y\}, \; C_t[a] = C_t'[a]$;
    \item[(c)] $C_t[z] = C_t'[z'] = \min_{a \in \mathcal{T}_t} C_t[a] = \min_{a \in \mathcal{T}_t'} C_t'[a]$.  
  \end{itemize}
  \medskip\hrule\medskip
  \item[$\sfour$]
  \begin{itemize}
    \item[(a)] $\mathcal{T}_t \neq \mathcal{T}_t'$; $W_t = \{z_1, z_2\}, \; W_t' = \{z_1', z_2'\}$;
    \item[(b)] $\forall a \in \mathcal{T}_t \cap \mathcal{T}_t', \; C_t[a] = C_t'[a]$;
    \item[(c)] $C_t[z_1]-1 = C_t[z_2] = C_t'[z_1'] = C_t'[z_2'] = \min_{a \in \mathcal{T}_t} C_t[a] = \min_{a \in \mathcal{T}_t\pr} C_t\pr[a]$; 
    \item[(d)] $\arg \max_{a \in S_{t+1}'} \tau_t'[a] \in W_t'$.
  \end{itemize}
\end{itemize}
\end{topbox}
\caption{State data structure differences for $\spacesaving$ on neighboring streams.}
\label{tab:states}
\end{table*}

A $\diffpriv$ linear sketch can be used as a frequency oracle within a heavy hitter search algorithm. 
The most direct approach performs a brute-force scan over the entire item universe~\cite{epasto2023differentially}, which is computationally intractable for large (which is to say data stream) domains.
More efficient methods employ tree-based search procedures that recursively prune infrequent regions of the domain~\cite{bassily2017practical}.
Nevertheless, these approaches still require $\mathcal{O}(\sqrt{d})$ oracle queries, where $d$ denotes an upper bound on the number of distinct items, resulting in non-trivial (and non-streaming) computational and memory overheads.

Alternatively, in the non-private setting, heavy hitters can be extracted from a linear sketch by tracking the items with the largest observed frequency estimates throughout the stream.
However, as with counter-based algorithms, the resulting candidate sets may differ across neighboring streams, rendering their direct release incompatible with $\diffpriv$.
Designing a privacy-preserving analogue of this tracking-based approach is non-trivial, and addressing this challenge is a central contribution of our work.

\section{Differentially Private SpaceSaving}
\label{sec:soss}

To privatize a counter-based algorithm it is not sufficient to inject noise into the counts of the tracked items.
This is because the identities of the tracked items can differ on neighboring inputs, allowing an adversary to easily distinguish the streams based on the output of the algorithm.
To overcome this challenge for $\mg$, Lebeda and Tetek enumerate all possible \textit{differences} between the structures $\mgx$ and $\mgxp$, for $X \sim X^{\prime}$ \cite{lebeda2023better}.
Knowledge of these differences allows a non-private $\mg$ to be processed such that both privacy leakage events are suppressed and item count differences are obscured. 
We adopt the same approach for privatizing $\spacesaving$.

\subsection{SpaceSaving on Neighboring Streams}

Before proceeding, we provide some useful notation.
Let $X = \{x_1, \ldots, x_T  \}$ and $X^{\prime} = \{x_1, \ldots, x_{i-1}, \varnothing, x_{i+1}, \ldots, x_T\}$ denote neighboring inputs, let $(\mathcal{T}_t, C_t) \gets \ssx(X_t)$ denote the output of $\spacesaving$ at time $t$, and let $W_t = \mathcal{T}_t\setminus \mathcal{T}_t^{\prime}$ and $W_t\pr = \mathcal{T}_t\pr \setminus \mathcal{T}_t$  denote the set difference functions.
Item eviction is determined by the recency functions  $\tau_t[y] = \max \{ r \mid X[r] = y,\; r \le t \}$ and $\tau_t\pr[y] = \max \{ r \mid X\pr[r] = y,\; r \le t \}$.
Define the sets of items with minimum count, prior to eviction, as $S_{t+1} = \arg \min \{C_t[z] \mid z \in \mathcal{T}_{t}\}$ and $S_{t+1}\pr = \arg \min \{C_t\pr[z] \mid z \in \mathcal{T}_{t}\pr\}$.
Thus, our eviction strategy at time $t$ is $\arg \max_{a \in S_t} \tau_{t-1}[a]$.
We refer to this item as occupying the \textit{eviction register}.

To formalize the differences between the outputs of the $\spacesaving$ algorithm on neighboring streams, we categorize each {type} of difference into a \emph{state}.
Our main privacy proof requires the enumeration of all possible states, which are provided in Table~\ref{tab:states}.
The following lemma establishes that this enumeration is exhaustive.
\begin{lemma}
    For any $k \in \mathbb{Z}^+$ and $t \ge i$, let
    \begin{align*}
        \ssx(X_t = \{x_1, \ldots, x_t\}, k) 
        &= (\mathcal{T}_t, C_t) \\
        \ssx(X_t' = \{x_1, \ldots, x_{i-1}, \varnothing, x_{i+1}, \ldots, x_t\}, k) 
        &= (\mathcal{T}_t', C_t').
    \end{align*}
    Then, the pair of outputs $(\mathcal{T}_t, C_t)$ and $(\mathcal{T}_t', C_t')$ can be in any of the states $\sone, \stwo, \sthree$, or $\sfour$ (Table~\ref{tab:states}).
    \label{lem:state_eos}
\end{lemma}

Given that the data structure differences are restricted to states $\sone-\sfour$, we obtain the following corollary that provides a high level generalization of the states.

\begin{corollary}
\label{cor:state_difference}
Let $\ssx(X, k) = (\mathcal{T}, C)$ and $\ssx(X', k) = (\mathcal{T}', C')$ be the outputs of Algorithm~\ref{alg:SS} on neighboring streams $X$ and $X'$, where $X'$ is obtained from $X$ by removing a single update.
Then the following properties hold:
\begin{enumerate}
    \item The intersection satisfies $|\mathcal{T} \cap \mathcal{T}'| \ge k-2$.
    \item For every item $a \in \mathcal{T} \setminus (\mathcal{T} \cap \mathcal{T}')$, its counter obeys
    \[
        C[a] \le \min_{b \in \mathcal{T}} C[b] + 1,
    \]
    and for every item $a \in \mathcal{T}' \setminus (\mathcal{T} \cap \mathcal{T}')$,
    \[
        C'[a] \le \min_{b \in \mathcal{T}'} C'[b].
    \]
    \item There exists at most one item $x \in \mathcal{T}\cap \mathcal{T}'$ such that $C[x] = C'[x] + 1$, and for all $y \in (\mathcal{T} \cap \mathcal{T}') \setminus \{x\}$, we have $C[y] = C'[y]$.
\end{enumerate}
\end{corollary}
\noindent
This corollary informs the design of a $\diffpriv$ post-processing procedure for the output of the non-private $\spacesaving$ algorithm.
Before outline such a procedure, we turn to the formal justification of this claim by establishing Lemma~\ref{lem:state_eos}.

The proof of Lemma~\ref{lem:state_eos} proceeds by induction on the stream length.
We break the proof up into a sequence of intermediate lemmas for readability and begin with the base case (time $t=i$).

\begin{lemma}
    For any $k \in \mathbb{Z}^+$ and time $t=i$, the outputs
    \begin{align*}
        \ssx(\{x_1, \ldots, x_i  \}, k) &= \mathcal{T}_i, C_i \\
        \ssx(\{x_1, \ldots, x_{i-1}, \varnothing\}, k) &= \mathcal{T}_i\pr, C_i\pr
    \end{align*} 
    are in either $\sone$ or $\stwo$.
    \label{lem:init_so_ss}
\end{lemma}

\begin{figure}
    \centering
    \includegraphics[width=0.9\linewidth]{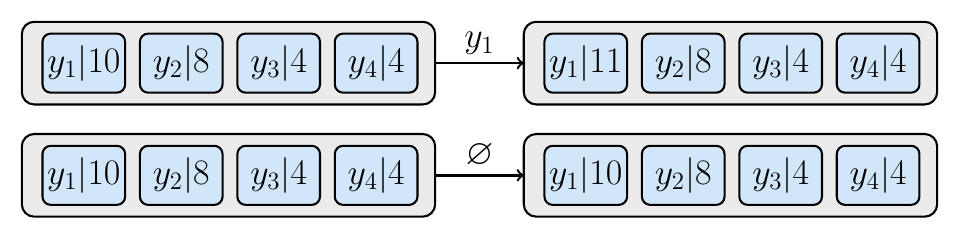}
    \caption{
    Example of transition to $\sone$ when element $y_1 \in \mathcal{T}$ arrives from $X$ and $X\pr$ provides no update.
    %Isolated elements are highlighted in red.
    }
    \label{fig:s0_s1}
\end{figure}
\begin{figure}
    \centering
    \includegraphics[width=0.9\linewidth]{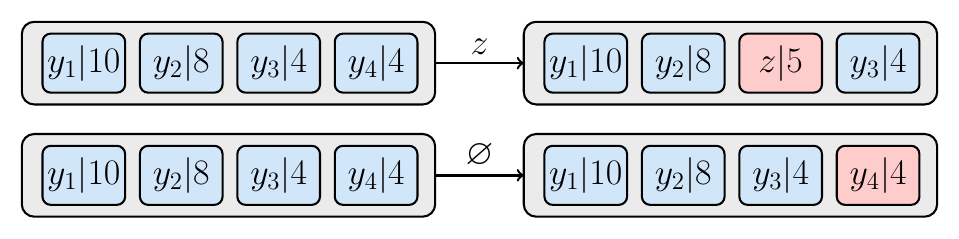}
    \caption{
    Example of transition to $\stwo$ when element $z \notin \mathcal{T}$ arrives from $X$ and $X\pr$ provides no update.
    Isolated elements are highlighted in red.
    }
    \label{fig:s0_s2}
\end{figure}

\begin{proof}
    There are two cases to consider for update $i$.
    First, we have $x_i \in \mathcal{T}_{i-1} (= \mathcal{T}_{i-1}\pr)$.
    In this instance the counter $C_{i-1}[x_i]$ gets incremented to form $C_i[x_i]$, and the set of tracked items stays the same (condition $\sone$(a)).
    As $C_{i-1}\pr = C_{i}\pr$, conditions $\sone$(b) and $\sone$(c) are satisfied.
    
    Second, we observe $x_i \notin \mathcal{T}_{i-1}$.
    Here,  the item $y = \arg \max_{a \in S_t} \tau_{t-1}[a]$ is replaced by $x_i$ in $\mathcal{T}_{i-1}$ to form $\mathcal{T}_i$, leading to $\mathcal{T}_i \neq \mathcal{T}_i\pr$.
    Further, $x_i$ is allocated the count $C_i[x_i] = C_{i-1}[y] + 1 = C_i\pr[y]+1$.
    As $W_i = \{x_i\}$ and $W\pr = \{y\}$, with $y$ obtaining the minimum count in $\mathcal{T}_t\pr$,
    this results in a transition to $\stwo$. 
\end{proof}

Examples of the possible state transitions at time $i$ are given in Figures~\ref{fig:s0_s1} \&~\ref{fig:s0_s2}.
Before proceeding with the state differences at times $t>i$, we  prove a simple but useful result.
\begin{lemma}
    For all $t \geq i$ and $a \in \mathcal{T}_t \cap \mathcal{T}_t\pr$,
    \[
        C_t[a] = C_t\pr[a] \implies \tau_t[a] = \tau_t\pr[a].
    \]
    \label{lem:tau_implication}
\end{lemma}
\begin{proof}
    By Lemma~\ref{lem:init_so_ss}, at $t=i$, according to both $\sone$ and $\stwo$, $\forall a \in \mathcal{T}_t \cap \mathcal{T}_t\pr \setminus\{x_i\}$ we observe $C_t[a] = C_t\pr[a]$.
    As, $\forall a \in \mathcal{U}\setminus \{x_i\}$ and $t>0$, $\tau_t[a] = \tau_t\pr[a]$, the implication holds at $t=i$.
    At times $t>i$, all updates are identical. 
    Therefore, existing tracked items with identical counts across streams maintain identical counts \emph{and} recencies at subsequent arrivals.
    Further, any new items added to both tracked item sets with the same count also have the same recency.
\end{proof}

At times $t> i$, the remainder of the stream updates are identical for $X$ and $X^{\prime}$.
Significantly, the data structures do not deviate from states $\sone$-$\sfour$ as the streams proceed.

\begin{figure}
    \centering
    \includegraphics[width=0.9\linewidth]{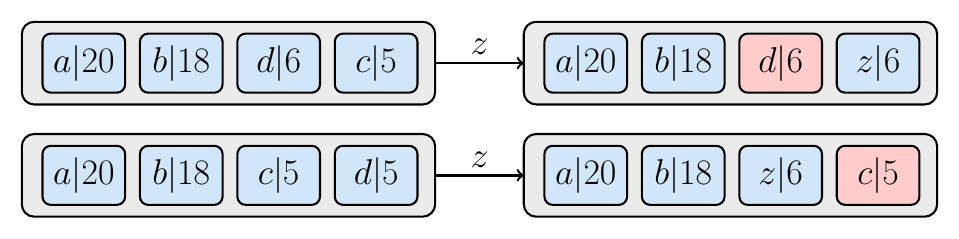}
    \caption{Example of a transition from $\sone$ to $\stwo$ when an element $z \notin \mathcal{T} \cap \mathcal{T}\pr$ arrives in the stream.}
    \label{fig:s1_s2}
\end{figure}
\begin{lemma}
    If the process is in state $\sone$ at time $t-1$, it can be states $\sone$ or $\stwo$ at time $t$.
    \label{lem:sone}
\end{lemma}
\begin{proof}
    If the process is in $\sone$, we have to consider the two cases $x_{t} \in \mathcal{T}_{t-1} (=\mathcal{T}_{t-1}\pr)$ and $x_{t} \notin \mathcal{T}_{t-1}$.
    In the former case, the sets $\mathcal{T}_{t-1}$ and $\mathcal{T}_{t-1}\pr$ do not change and both $C_{t-1}[x_t]$ and $C_{t-1}\pr[x_t]$ are incremented to form $C_t$ and $C_{t}\pr$. 
    Thus, the process stays in $\sone$.
    
    In the latter case, $x_{t} \notin \mathcal{T}_{t-1}$, items are evicted from both sets of tracked items and replaced with $x_t$.
    Recall that there exists $y \in \mathcal{T}_{t-1}=\mathcal{T}_{t-1}\pr$ such that $C_{t-1}[y] = C_{t-1}\pr[y] + 1$.
    There are two possibilities to consider based on the identities of the items in the eviction registers.
    \begin{itemize}
        \item[(i)] $\arg \max_{a \in S_t} \tau_{t-1}[a] = \arg \max_{a \in S_t\pr} \tau_{t-1}\pr[a] $. 
        In this instance, identical items are removed from both sets, resulting in the update $\mathcal{T}_t = \mathcal{T}_t\pr$ (condition $\sone$(a)). 
        If $x_t$ replaces $y$ in both sets, we observe $C_t[x_t] = C_t\pr[x_t]+1$, otherwise we observe $C_t[x_t] = C_t\pr[x_t]$, preserving the conditions $\sone$(b)-(c).  
        \item[(ii)] $\arg \max_{a \in S_t} \tau_{t-1}[a] \neq \arg \max_{a \in S_t\pr} \tau_{t-1}\pr[a] $.
        Let 
        \[
            x^* = \arg \max_{a \in S_t}\tau_{t-1}[a]
        \]
        By Lemma~\ref{lem:tau_implication},
        this case occurs if and only if $y \notin S_t$ and $y = \arg \max_{a \in S_{t}\pr} \tau_{t-1}\pr[a]$.
        This also implies that $C_{t-1}[y] = C_{t-1}[x^*] + 1$. 
        The following tracked item replacements occur 
        \begin{align*}
            \mathcal{T}_{t} &\gets (\mathcal{T}_{t-1}\setminus \{x^*\}) \cup \{x_t\}, \\
            \mathcal{T}_{t}\pr &\gets (\mathcal{T}_{t-1}\pr\setminus \{y\}) \cup \{x_t\}.
        \end{align*}
        This leads to $\mathcal{T}_t \neq \mathcal{T}_t\pr$ (condition $\stwo$(a)).
        Item $x_t$ is allocated the count $C_t[x_t] = C_{t-1}[x^*] +1 = C_{t-1}\pr[y] + 1 = C_t\pr[x_t]$ in both structures.
        As the counts of $\mathcal{T}_t \cap \mathcal{T}_t\pr \setminus \{x_t\}$ (which does not include $y$) remain unchanged, condition $\stwo$(b) is satisfied.
        In addition, we observe $ W_t = \{y\}$ with $C_t[y] = C_{t-1}[y] = C_{t-1}[x^*]+1$ (as stated above) and $W_t\pr =\{x^*\}$  with $C_{t}\pr[x^*] = C_{t-1}\pr[x^*] = C_{t-1}[x^*]$.
        Thus, $C_t[y] = C_{t}\pr[x^*] + 1$ (condition $\stwo$(c)).
        Lastly, 
        $x^*$ occupying the eviction register in $\mathcal{T}_{t-1}$ implies that $x^*$ has the minimum count in $\mathcal{T}_{t-1}\pr\setminus\{y\}$.
        Thus, as $x_t$ is added with count $C_{t-1}\pr[y]+1 > C_{t-1}[x^*] = C_{t}\pr[x^*]$, $x^*$ obtains the minimum count in $\mathcal{T}_t\pr$.
        This results in a transition to $\stwo$.
    \end{itemize}
\end{proof}

An illustration of a transition from $\sone$ to $\stwo$ is available in Figure~\ref{fig:s1_s2}.
A transition out of $\sone$ at time $t$ can only occur when item $y\in \mathcal{T}_{t-1}\cap \mathcal{T}_{t-1}\pr$ such that $\mathcal{C}_{t-1}[y] = \mathcal{C}_{t-1}\pr[y] +1$ is in the eviction register for $\mathcal{T}_{t-1}\pr$ but not $\mathcal{T}_{t-1}$.

\begin{restatable}{lemma}{lemstwo}
    If the process is in state $\stwo$ at time $t-1$, it can be states $\sone$, $\stwo$ or $\sthree$ at time $t$.
    \label{lem:stwo}
\end{restatable}

The proof is relegated to Appendix~\ref{app:proof}, due to its similarity to the proof of Lemma~\ref{lem:sone}.
Examples of state transitions from $\stwo$ are given in Figures~\ref{fig:s2_s1} \&~\ref{fig:s2_s3}.
The transition to $\sone$ occurs when the isolated elements of both structures are in the eviction register.
In contrast, the transition to $\sthree$ occurs when the isolated element in $\mathcal{T}_{t-1}$ is not in the eviction register.
Both transitions can be triggered when an item $y\notin \mathcal{T}_{t-1} \cup \mathcal{T}_{t-1}\pr$ arrives.

\begin{figure}
    \centering
    \includegraphics[width=0.9\linewidth]{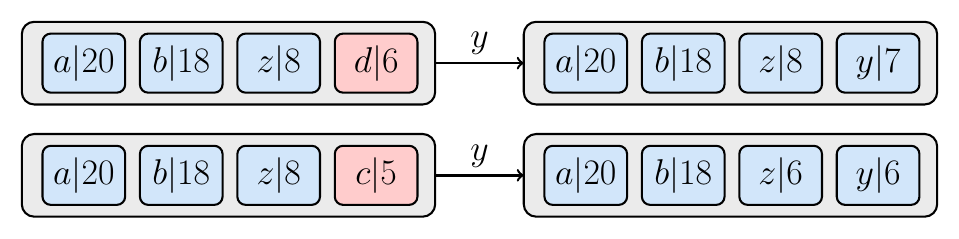}
    \caption{
    Example of a transition from $\stwo$ to $\sone$ when an element $y \notin \mathcal{T} \cap \mathcal{T}\pr$ arrives in the stream. 
    Note that this transition also occurs if either $d \in \mathcal{T}\setminus \mathcal{T}\pr$ or $c \in \mathcal{T}\setminus\mathcal{T}\pr$ arrive in this position.
    }
    \label{fig:s2_s1}
\end{figure}

\begin{figure}
    \centering
    \includegraphics[width=0.9\linewidth]{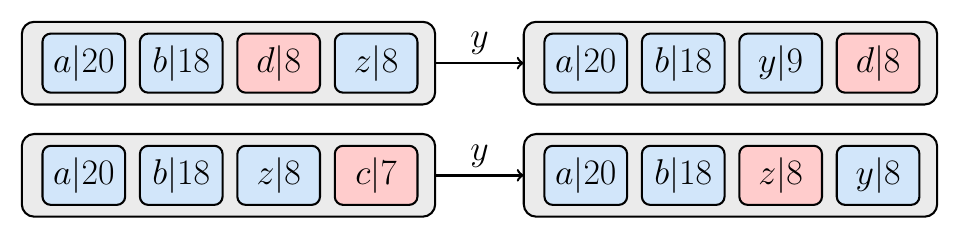}
    \caption{
    Example of a transition from $\stwo$ to $\sthree$ when an element $y \notin \mathcal{T} \cup \mathcal{T}\pr$ arrives in the stream. 
    %Note that this transition also occurs if either $d \in \mathcal{T}\setminus \mathcal{T}\pr$ or $c \in \mathcal{T}\setminus\mathcal{T}\pr$ arrive in this position.
    }
    \label{fig:s2_s3}
\end{figure}

\begin{restatable}{lemma}{lemsthree}
    If the process is in state $\sthree$ at time $t-1$, it can be states $\sone$, $\sthree$ or $\sfour$ at time $t$.
    \label{lem:sthree}
\end{restatable}

The proof is relegated to Appendix~\ref{app:proof}.
An example of a transition from $\sthree$ to $\sone$ is given in Figure~\ref{fig:s3_s1}.
The transition is triggered if an isolated element is in the eviction register of $\mathcal{T}_{t-1}$ but not $\mathcal{T}_{t-1}\pr$.
An example of a transition from $\sthree$ to $\sfour$ is given in Figure~\ref{fig:s3_s4} and increases the number of isolated elements.
The transition is triggered if the eviction registers are not identical and do not contain isolated elements.
%Due to the recency condition,
This occurs when the eviction register in $\mathcal{T}_{t-1}\pr$ is occupied by the unique element $y\in \mathcal{T}_{t-1}\cap \mathcal{T}_{t-1}\pr$ where $C_{t-1}[y] = C_{t-1}\pr[y]+1$.

\begin{figure}
    \centering
    \includegraphics[width=0.9\linewidth]{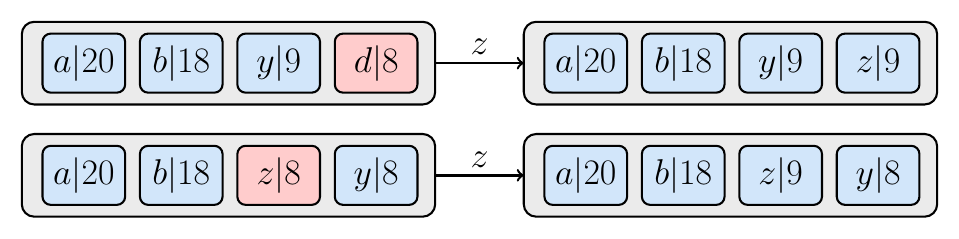}
    \caption{
    Example of a transition from $\sthree$ to $\sone$ when an element $z \in \mathcal{T}\pr \setminus \mathcal{T}$ arrives in the stream. 
    %Note that this transition also occurs if either $d \in \mathcal{T}\setminus \mathcal{T}\pr$ or $c \in \mathcal{T}\setminus\mathcal{T}\pr$ arrive in this position.
    }
    \label{fig:s3_s1}
\end{figure}

\begin{figure}[t]
    \centering
    \includegraphics[width=0.9\linewidth]{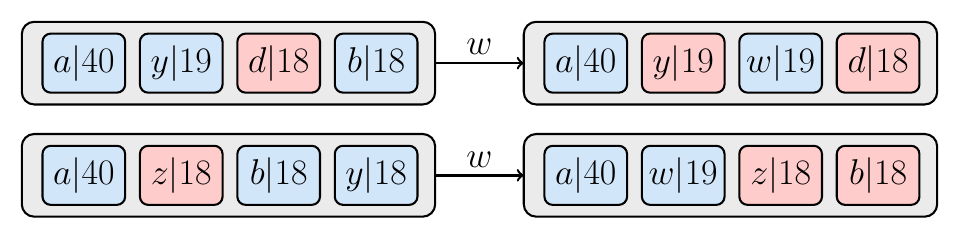}
    \caption{
    Example of a transition from $\sthree$ to $\sfour$ when an element $w \notin \mathcal{T}\cup\mathcal{T}\pr$ arrives in the stream. 
    Note that the eviction elements are (i) not isolated elements and (ii) not identical.
    %Note that this transition also occurs if either $d \in \mathcal{T}\setminus \mathcal{T}\pr$ or $c \in \mathcal{T}\setminus\mathcal{T}\pr$ arrive in this position.
    }
    \label{fig:s3_s4}
\end{figure}

\begin{restatable}{lemma}{lemsfour}
    If the process is in state $\sfour$ at time $t-1$, it can be states $\stwo$, $\sthree$ or $\sfour$ at time $t$.
    \label{lem:sfour}
\end{restatable}
The proof is relegated to Appendix~\ref{app:proof}.
An important condition of $\sfour$ is that isolated elements from both streams are always in the eviction registers.
This restricts available transitions to $\stwo$ or $\sthree$ only, as, if an element outside the intersection $\mathcal{T}_{t-1}\cap \mathcal{T}_{t-1}\pr$ arrives, the number of isolated elements must decrease.
The resulting state depends on the identify of the element and examples are given in Figures~\ref{fig:s4_s2} \&~\ref{fig:s4_s3}.

We are now in a position to prove the main result.
\begin{proof}[Proof of Lemma~\ref{lem:state_eos}]
    Proceeding inductively, Lemma~\ref{lem:init_so_ss} provides the base case and Lemmas~\ref{lem:sone}-\ref{lem:sfour} provide the subsequent inductive steps.
\end{proof}

\subsection{Privacy}

\begin{figure}
    \centering
    \includegraphics[width=0.9\linewidth]{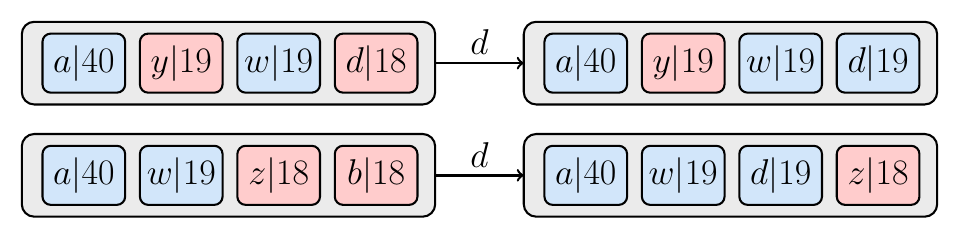}
    \caption{
    Example of a transition from $\sfour$ to $\stwo$ when an element $d \in \mathcal{T} \setminus \mathcal{T}\pr$, where $C[d]= \min_{x \in \mathcal{T}} C[x]$,  arrives in the stream. 
    %Note that this transition also occurs if either $d \in \mathcal{T}\setminus \mathcal{T}\pr$ or $c \in \mathcal{T}\setminus\mathcal{T}\pr$ arrive in this position.
    }
    \label{fig:s4_s2}
\end{figure}

\begin{figure}
    \centering
    \includegraphics[width=0.9\linewidth]{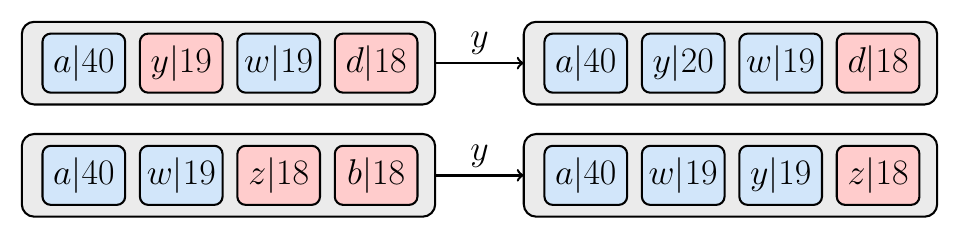}
    \caption{
    Example of a transition from $\sfour$ to $\sthree$ when an element $y \in \mathcal{T} \setminus \mathcal{T}\pr$, where $C[y]= \min_{x \in \mathcal{T}} C[x] +1$,  arrives in the stream. 
    %Note that this transition also occurs if either $d \in \mathcal{T}\setminus \mathcal{T}\pr$ or $c \in \mathcal{T}\setminus\mathcal{T}\pr$ arrive in this position.
    }
    \label{fig:s4_s3}
\end{figure}

The technical foundation for privatizing $\spacesaving$ is captured by Corollary~\ref{cor:state_difference}, 
which offers two pivotal insights into the algorithm’s stability under neighboring inputs.
First, consider the outputs
\begin{align*}
    (\mathcal{T}, C) &\gets \ssx(X, k) \\
    (\mathcal{T}\pr, C\pr) &\gets \ssx(X\pr, k),
\end{align*}
where $X'$ is derived from $X$ by the removal of a single element. 
Corollary~\ref{cor:state_difference}  asserts that among the items common to both tracked sets, that is, $x \in \mathcal{T} \cap \mathcal{T}'$, there exists at most one for which the counters differ, and this discrepancy is bounded by a single increment.

Second, it guarantees that items outside the intersection are associated with low counts, specifically bounded above by the minimum counter  plus one.
These properties together imply that the output of $\spacesaving$ is inherently stable, and thus amenable to privatization via two mechanisms: (1) obscuring the small count differences among shared items, and (2) suppressing isolated items whose counts fall below a threshold.

The first mechanism is realized by injecting $\laplacedist$ noise into the counters of tracked items. Since the sensitivity of the counters restricted to $\mathcal{T} \cap \mathcal{T}'$ is at most one, Lemma~\ref{lem:laplace_vector} ensures that adding independent noise drawn from $\laplacedist(1/\varepsilon)$ suffices to achieve differential privacy.

The second step can be achieved by applying a threshold test that suppresses items with low count.
Let $\gamma$ denote, with probability $1-\delta$, the maximum deviation introduced by $\laplacedist$ noise.
In $\spacesaving$, for a register of length $k$, the minimum count is at most $|X|/k$.
Therefore, it would be sufficient to set the threshold to $|X|/k+1+\gamma$.
However, as our task is to detect items whose frequency is at least $|X|/k$ (heavy hitters), this introduces a tension.
Notably, overly aggressive suppression risks excluding valid heavy hitters that can be reported without breaching privacy.

To reconcile privacy with utility, we initialize a non-private $\spacesaving$ data structure with an expanded capacity $\tilde{k} > k$.
This adjustment ensures that items with minimal counts are unlikely to be heavy hitters, thereby allowing suppression without sacrificing recall. 
Subsequently, it is sufficient to set the threshold to 
\[
    \tau \gets \max \{|X|/k - \gamma,  |X|/\tilde{k} + 1 + \gamma\}.
\]
With high probability, the term $|X|/k - \gamma$ guarantees the admission of heavy hitters and the term $|X|/\tilde{k} + 1 + \gamma$ suppresses isolated items.
This strategy is formalized in Algorithm~\ref{alg:DPSS}.
We now prove the privacy properties.

\begin{algorithm}[t]
  \SetAlgoLined
  \DontPrintSemicolon
  \SetKwProg{myproc}{define}{}{}
  \SetKwInOut{require}{let}
  %\require{$\tau_t[a] = \max \{\, t' \mid x_{t'} = a,\; t' \le t \,\}$ }
    \myproc{$\mathsf{DPSS}(X, k, \tilde{k}, \varepsilon, \delta)$}{
        $\gamma \gets \tfrac{1}{\varepsilon}\ln\!\left(\tfrac{2}{\delta}\right)$ \;
        $\mathcal{T}, C \gets \ssx(X, \tilde{k})$\;
        $\hat{\mathcal{T}}, \hat{C} \gets \varnothing$\;
        $\tau \gets \max \{|X|/k - \gamma,  |X|/\tilde{k} + 1 + \gamma\}$\;
        \For{$x \in \mathcal{T}$}
        {
            $\hat{c}_x \gets C[x] + \laplacedist(1/\varepsilon)$\;
            \If{$\hat{c}_x > \tau$}
            {
                $\hat{\mathcal{T}} \gets \hat{\mathcal{T}}\cup \{x\}$\;
                $\hat{C}[x] \gets \hat{c}_x$\;
            }
        }
        \KwRet $\hat{\mathcal{T}}, \hat{C}$\;
    }
  \caption{Differentially Private $\spacesaving$}
  \label{alg:DPSS}
\end{algorithm}

\begin{theorem}
    For any $\varepsilon>0$ and $\delta \in (0,1)$, Algorithm~\ref{alg:DPSS} is $(\varepsilon, \delta)$-differenitally private.
    \label{thm:priv_ss}
\end{theorem}

\begin{proof}
For any fixed set of input parameters $(k, \tilde{k})$ define
\[
    \mathcal{A}(X) = \mathsf{DPSS}(X, k, \tilde{k}, \varepsilon, \delta).
\]
We prove that for any pair of neighboring inputs $X,X'$ (where $X'$ is obtained by removing a single element from $X$), and any measurable subset $S$ of outputs,
\[
\Pr[\,\mathcal{A}(X)\in S\,] \;\le\; e^{\varepsilon}\,\Pr[\,\mathcal{A}(X')\in S\,] \;+\; \delta.
\]

Define the ``stable selection'' event $E$ as the event that no item in the symmetric difference
\(
U := \mathcal{T}\triangle\mathcal{T}'
\)
exceeds the threshold in either run.
That is,
\[
E \;:=\;
\bigl[\,
\forall u\in U\cap \mathcal{T}:\; C[u]+Z_u \le \tau
\ \text{ and }\ 
\forall u\in U\cap \mathcal{T}':\; C'[u]+Z'_u \le \tau
\,\bigr],
\]
where $Z_u$ is the Laplace noise.
As $\spacesaving$ is composed of $\tilde{k}$ counters that sum to $|X|$, the minimum count is at most $|X|/\tilde{k}$.
Therefore, by Corollary~\ref{cor:state_difference} and Lemma~\ref{lem:ss_bound}, $|U|\le 2$ and each such $u\in U$ has \emph{true} count at most $|X|/\tilde{k}+1$.
Since $\tau \ge |X|/\tilde{k} + 1 + \gamma$, the margin from the minimum count to the threshold is at least $\gamma$ for every such $u \in U$.
For Laplace noise $Z\sim \mathrm{Lap}(1/\varepsilon)$ we have the standard tail bound
\(
\Pr[\,Z > \gamma\,] = \tfrac{1}{2}e^{-\varepsilon\gamma}.
\)
Applying a union bound over at most two unstable labels and two independent runs yields
\[
\Pr[\,\neg E\,]
\;\le\;
2 \cdot 2 \cdot \tfrac{1}{2}\, e^{-\varepsilon\gamma}
\;=\;
2\, e^{-\varepsilon\gamma}
\;\le\;
\delta,
\]
where the last inequality uses $\gamma = \frac{1}{\varepsilon}\ln\!\left(\frac{2}{\delta}\right)$.

Conditioned on the event $E$,
the released index sets $\hat{\mathcal{T}}$ and $\hat{\mathcal{T}}'$ are contained in $\mathcal{T}\cap \mathcal{T}'$.
%That is, no unstable index is released in either run.
Let $m := |\mathcal{T}\cap \mathcal{T}'|\ge \tilde{k}-2$, and fix an arbitrary but common ordering of the labels in $\mathcal{T}\cap \mathcal{T}'$.
Define the vector-valued function
\(
F(X) \in \mathbb{R}^m
\)
as the ordered list of common counters
\(
F(X) := (\, C[x] \,)_{x\in\mathcal{T}\cap\mathcal{T}'},
\)
and similarly $F(X') := (\, C'[x] \,)_{x\in\mathcal{T}\cap\mathcal{T}'}$.
By Corollary~\ref{cor:state_difference}, $F(X)$ and $F(X')$ differ by at most $1$ in a single coordinate and are equal elsewhere. 
Thus, 
\(
\|F(X)-F(X')\|_1 \le 1.
\)

Let $Z=(Z_x)_{x\in\mathcal{T}\cap\mathcal{T}'}$ and $Z'=(Z'_x)_{x\in\mathcal{T}\cap\mathcal{T}'}$ be noise vectors with $Z_x,Z'_x\sim \mathrm{Lap}(1/\varepsilon)$, independent across coordinates.
Then, by Lemma~\ref{lem:laplace_vector}, for all measurable $A\subseteq \mathbb{R}^m$,
\begin{equation}\label{eq:laplace-mech}
    \Pr[\,F(X) + Z \in A\,] \;\le\; e^{\varepsilon}\, \Pr[\,F(X\pr) + Z\pr \in A\,].
\end{equation}

Now define the (deterministic) post-processing map $\Phi$
that keeps precisely those coordinates of $F(X) + Z$ whose value exceeds $\tau$, together with their labels.
Conditioned on $E$, 
by the post-processing property of differential privacy, applying $\Phi$ preserves Inequality~\eqref{eq:laplace-mech}.
That is, for all measurable $S,$
\begin{equation}\label{eq:cond-dp}
    \Pr[\, \Phi(F(X) + Z) \in S \mid E \,]
    \;\le\;
    e^{\varepsilon}\,
    \Pr[\, \Phi(F(X\pr) + Z\pr) \in S \mid E \,].
\end{equation}

Condition on $E$, we have $\mathcal{A}(X) = \Phi(F(X) + Z)$.
Putting everything together,
using~\eqref{eq:cond-dp}, $\Pr[E]\le 1$, and $\Pr[\neg E]\le \delta$, we obtain
\begin{align*}
\Pr[\,\mathcal{A}(X)\in S\,]
&= \Pr[\,\mathcal{A}(X)\in S \mid E\,]\,\Pr[E] \\
&+ 
\Pr[\,\mathcal{A}(X)\in S \mid \neg E\,]\,\Pr[\neg E] \\
&\le
e^{\varepsilon}\,\Pr[\,\mathcal{A}(X')\in S \mid E\,]\,\Pr[E] \\
&+
\Pr[\,\mathcal{A}(X)\in S \mid \neg E\,]\,\Pr[\neg E]
\\
&\le
e^{\varepsilon}\,\Pr[\,\mathcal{A}(X')\in S\, \mid E]\Pr[E]
\;+\;
\Pr[\neg E]
\\
&\le
e^{\varepsilon}\,\Pr[\,\mathcal{A}(X')\in S\,]
\;+\;
\Pr[\neg E] \\
&\le
e^{\varepsilon}\,\Pr[\,\mathcal{A}(X')\in S\,]
\;+\;
\delta \\
\end{align*}
The same analysis holds when exchanging the roles of $X$ and $X'$.
Thus, Algorithm~\ref{alg:DPSS} satisfies $(\varepsilon,\delta)$-differential privacy.
\end{proof}

\subsection{Performance and Utility}

Algorithm~\ref{alg:DPSS} inherits the efficiency of the underlying $\spacesaving$ algorithm, offering strong performance guarantees in both time and space.
Using an implementation based on hash tables and doubly linked lists,  Algorithm~\ref{alg:DPSS} can be implemented with $O(\tilde{k})$ space complexity and can process each stream update in $O(1)$ time \cite{cormode2008finding}. 
This ensures that the algorithm scales gracefully with high-throughput data streams.

In practice, for large streams, the bound on the minimum counter value $|X|/\tilde{k}$ dominates the noise parameter $\gamma = \frac{1}{\varepsilon}\ln\left(\frac{2}{\delta}\right)$. 
Consequently, it is sufficient to set $\tilde{k} = O(k)$ without compromising utility. 
Thus, the additional space does not come at an asymptotical cost.

For example, consider a stream of length $|X| = 2^{28}$, with $k = 512$, and privacy parameters $\varepsilon = 0.1$ and $\delta = 10^{-3}$.
We subsequently choose $\tilde{k} > k$ to ensure that the suppression threshold does not exclude true heavy hitters.
Recall the threshold used by Algorithm~\ref{alg:DPSS}:
\[
\tau \;=\; \max\Bigl\{\, |X|/k - \gamma,\; |X|/\tilde{k} + 1 + \gamma \,\Bigr\},
\quad\text{where}\quad
\gamma \;=\; \frac{1}{\varepsilon}\ln\!\left(\frac{2}{\delta}\right).
\]
To ensure $\tau = |X|/k - \gamma$, so heavy hitters are not suppressed, we require
\[
\frac{|X|}{\tilde{k}} + 1 + 2\gamma \le \frac{|X|}{k}
\quad\Rightarrow\quad
\tilde{k} \ge \frac{|X|}{\frac{|X|}{k} - (1 + \gamma)}.
\]
Plugging in values, with $\gamma = 89.87$, we get that it is sufficient to set $\tilde{k}$ to 513.

These observations are summarized in the following result.
\begin{restatable}{theorem}{thmutility}
    For any $\varepsilon>0$ and $\delta\in(0,1)$, let $\gamma=\tfrac{1}{\varepsilon}\ln\left(\tfrac{2}{\delta}\right)$.
    Suppose the input stream $X$ and target capacity $k$ satisfy
    \[
        |X|/(2k)>2(\gamma+1).
    \]
    Then Algorithm~\ref{alg:DPSS} can be implemented with $\mathcal{O}(k)$ words of memory and $\mathcal{O}(1)$ update time.
    For any item $x\in\hat{\mathcal{T}}$ with true frequency $f_x$, with probability at least $1-\delta$, the released estimate $\hat{C}[x]$ obeys
    \[
    f_x \;-\; \tfrac{1}{\varepsilon}\ln\tfrac{1}{\delta}
    \;\le\; \hat{C}[x]
    \;\le\; f_x + \frac{|X|}{2k} + \tfrac{1}{\varepsilon}\ln\tfrac{1}{\delta}.
    \]
    In addition, any item $y$ with $f_y>|X|/k$ is released (that is, $y\in\hat{\mathcal{T}}$) with probability at least $1-\delta$.
\end{restatable}
\noindent
The proof is relegated to Appendix~\ref{app:sub_utility}.
Notably, there are inherent accuracy lower bounds for heavy-hitter mechanisms.
In the streaming setting, any method that outputs at most \(k\) items has a worst-case input on which the error is at least \(|X|/(k+1)\) \cite{bose2003bounds}. 
When differential privacy is imposed, the error must additionally scale as \(\mathcal{O}\!\left(\tfrac{1}{\varepsilon}\log\!\left(\tfrac{1}{\delta}\right)\right)\) for some inputs~\cite{balcer2017differential}.
Therefore, the error of our mechanism is asymptotically optimal.

\section{Frequent Items with a Frequency Oracle}
\label{sec:lin_sketch}

Linear sketches with $\diffpriv$ frequency estimates have been used as frequency oracles in  resource intensive (relative to the streaming model of computation) heavy hitter search algorithms. 
A simple way to bypass central server processing of frequency oracles would be to track a compact set of candidates online, akin to counter-based approaches.
However, this strategy is not automatically label-private.
Privacy leakage can occur, for neighboring streams $X = X\pr \cup \{y\}$, when $y$ appears as a tracked item in $X$ but not $X\pr$.
For example, imagine $y$ is the last item in stream $X$.
At time $|X|-1$, the tracked item sets have the same output distributions.
However, after the appearance of $y$ on $X$, the output distributions of the tracked item sets can differ significantly.

To mitigate this concern, 
we show that a simple wrapper, parameterized by a capacity $\tilde{k}>k$ and a single threshold tied to the oracle's error envelope, converts any $\diffpriv$ frequency estimator into a label-private heavy-hitter mechanism in the single-observation model.
The wrapper comes at an additional memory cost of $\mathcal{O}(\tilde{k})$.

\subsection{Overview}
The streaming data structure $\mathcal{D}$ underlying the wrapper requires the following two properties.
First, at time $t$, $\mathcal{D}$ provides frequency estimate $\tilde{f}_t(x)$ for stream arrival $x$, and admits the error bound
\[
\gamma_1(t)\ \ge\ \tilde{f}_t(x)-f_t(x)\ \ge\ \gamma_2(t),
\]
with probability $1-\delta/T$,
where $\gamma_1,\gamma_2:\mathbb{N}\to\mathbb{R}$ are non-decreasing in $t$. 
We refer to this error bound as the \emph{Error Envelope}.
Here $f_t(x)$ is the exact frequency of $x$ in the prefix of length $t$.
Second, all frequency estimates in $\mathcal{D}$ are $\varepsilon$-DP under single observation.

Similar to the $\spacesaving$ strategy, for target heavy hitter capacity $k$, we choose $\tilde{k}>k$ to isolate low-count, unstable labels.
We maintain, throughout the stream, a tracked set $\mathcal{T}$ of at most $\tilde{k}$ items whose arrival estimates are top-$\tilde{k}$. 
At release time $T$, we apply a single threshold
\[
\tau(T)\ :=\ \frac{T}{\tilde{k}}\ +\ \gamma_2(T)\ +\ 2\gamma_1(T) ,
\]
and return only those tracked items whose final estimate exceeds $\tau(T)$. 

Let $\mathcal{T}$ and $\mathcal{T}\pr$, respectively, denote the tracked item sets on neighboring $X$ and $X\pr$.
Intuitively, if $y$ appears in $\mathcal{T}$ but not $\mathcal{T}\pr$, then when $y$ last failed to be tracked on $X\pr$ (at some time $t<T$), its maximum \emph{true} frequency is bounded by $\frac{t}{\tilde{k}}+\gamma_2(t)+\gamma_1(t)\le \tau(T) + \gamma_1(T)$.
Therefore, unstable labels are suppressed by the threshold, while labels that are admitted by the threshold are protected by $\varepsilon$-DP through their frequency estimates.

\begin{algorithm}[t]
  \SetAlgoLined
  \DontPrintSemicolon
  \SetKwProg{myproc}{define}{}{}
  \SetKwInOut{require}{require}
  %\require{
    %$\mathcal{D}$ such that 
    %$\gamma_1(t)\ \ge\ \tilde{f}_t(x)-f_t(x)\ \ge\ \gamma_2(t)\qquad \forall x\in \mathcal{X}, t \in [T]$
   % }
    \myproc{$\mathsf{Top-}\tilde{k}(\mathcal{D}, X, \tilde{k}, \varepsilon)$}{
        %$\gamma \gets \tfrac{1}{\varepsilon}\ln\!\left(\tfrac{2}{\delta}\right)$ \;
        Initialize $\varepsilon$-DP frequency estimator $\mathcal{D}$\;
        $\mathcal{T}\gets \varnothing$\;
        $C \gets \varnothing$\;
        \For{$x \in X$}
        {
            $\tilde{f}_x \gets \mathcal{D}.\mathsf{update}(x)$\;
            \uIf{$x\in\mathcal{T}$}{
                $C[x] \gets \tilde{f}_x$\;
            }\uElseIf{$|\mathcal{T}|<\tilde{k}$}{
                $\mathcal{T} \gets \mathcal{T} \cup \{x\}$\;
            }\uElse{
                $y \gets \arg \min_{z \in \mathcal{T}} C[z]$\;
                \If{$\tilde{f}_x > C[y]$}{
                  $\mathcal{T} \gets (\mathcal{T}\setminus \{(y,m)\}) \cup \{(x,\tilde{f}_x)\} $\;
            }
          }
        }
        \KwRet ${\mathcal{T}}, C$\;
    }
  \caption{Tracking Top-$\tilde{k}$ frequent items}
  \label{alg:topk}
\end{algorithm}

\begin{algorithm}[t]
  \SetAlgoLined
  \DontPrintSemicolon
  \SetKwProg{myproc}{define}{}{}
  \SetKwInOut{require}{require}
  \require{
    $\mathcal{D}$ such that 
    $\gamma_1(t)\ \ge\ \tilde{f}_t(x)-f_t(x)\ \ge\ \gamma_2(t)$
    }
    \myproc{$\mathsf{EEHH}(\mathcal{D}, X, k, \tilde{k}, \varepsilon, \delta)$}{
        %$\gamma \gets \tfrac{1}{\varepsilon}\ln\!\left(\tfrac{2}{\delta}\right)$ \;
        $\mathcal{T}, C\gets \mathsf{Top-}\tilde{k}(\mathcal{D}, X, \tilde{k}, \varepsilon)$\tcp*{Algorithm~\ref{alg:topk}}
        $T\gets |X|$\;
        $\hat{\mathcal{T}}, \hat{C} \gets \varnothing$\;
        $\tau \gets \max\{T/k, T/\tilde{k} + 2\gamma_1(T) + \gamma_2(T)\}$\;
        \For{$x \in \mathcal{T}$}
        {
            $\tilde{f} \gets \mathcal{D}.\mathsf{query}(x)$\;
            \If{$C[x] > \tau \land \tilde{f} > \tau$}
            {
                $\hat{\mathcal{T}} \gets \hat{\mathcal{T}}\cup \{x\}$\;
                $\hat{C}[x] \gets \tilde{f}$\;
            }
        }
        \KwRet $\hat{\mathcal{T}}, \hat{C}$\;
    }
  \caption{Error Envelope Heavy Hitters}
  \label{alg:EEHH}
\end{algorithm}

\subsection{Privacy}

We now formalize the concept of suppressing of unstable labels on neighboring streams.

\begin{lemma}[Suppression of unstable labels]
\label{lem:suppress}
Let $\mathcal{D}$ be a data structure that estimates item frequency of stream arrival $x$ at time $t$ with the following error bounds
\[
\gamma_1(t)\ \ge\ \tilde{f}_t(x)-f_t(x)\ \ge\ \gamma_2(t).
\]
Let $X,X'$ be neighboring streams of length $T$, and let $\mathcal{T}, \mathcal{T}\pr$ be the tracked sets produced by Algorithm~\ref{alg:topk}. 
Fix $y\in \mathcal{T}\setminus \mathcal{T}\pr$. 
Let $t-1<T$ be the last time at which $y$ appeared in either $X\pr$ or $\mathcal{T}\pr$ (whichever is later). 
Then
\[
f_t(y)\ \le\ \frac{t}{\tilde{k}}+\gamma_2(t)+\gamma_1(t).
\]
The symmetric case $y\in \mathcal{T}\pr\setminus \mathcal{T}$ is analogous.
\end{lemma}

\begin{figure*}[t]
  \centering

  % ---- Row 1 ----
  \begin{subfigure}[t]{0.31\textwidth}
    \centering
    \includegraphics[width=\linewidth]{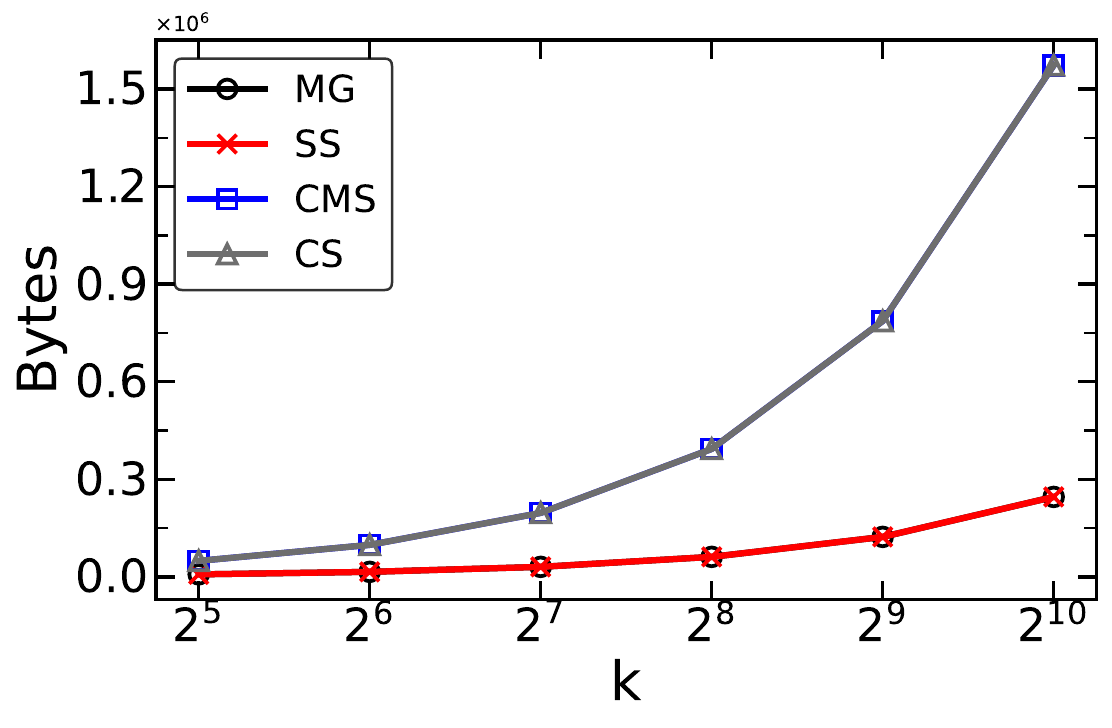}
    \caption{Memory vs $k$. Skew is $1.1$ and $\varepsilon = 0.1$.}
    \label{fig:syn_k_bytes}
  \end{subfigure}\hfill
  \begin{subfigure}[t]{0.31\textwidth}
    \centering
    \includegraphics[width=\linewidth]{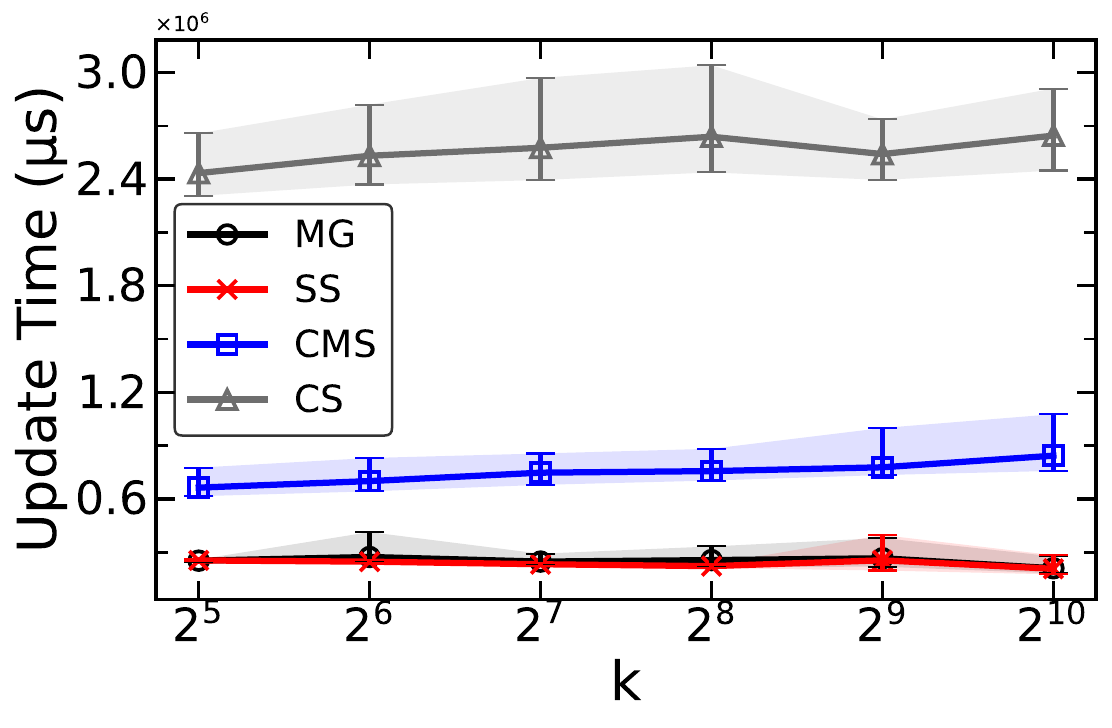}
    \caption{Time vs $k$. Skew is $1.1$ and $\varepsilon = 0.1$.}
    \label{fig:syn_k_update}
  \end{subfigure}\hfill
  \begin{subfigure}[t]{0.31\textwidth}
    \centering
    \includegraphics[width=\linewidth]{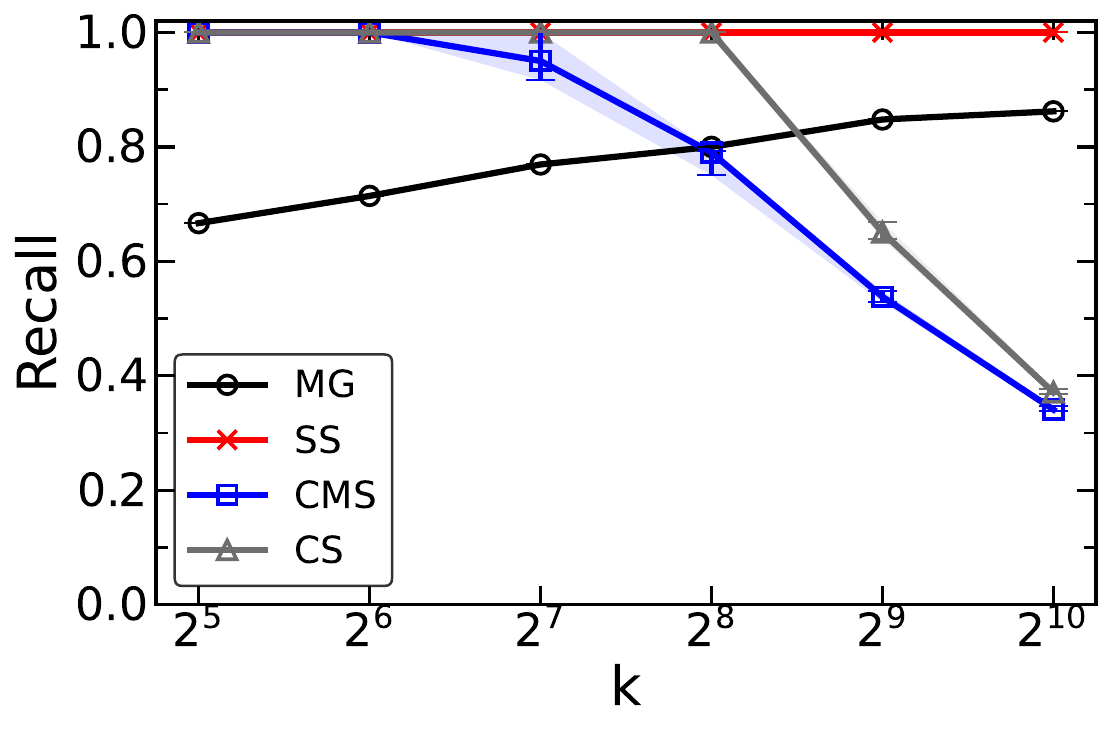}
    \caption{Recall vs $k$. Skew is $1.1$ and $\varepsilon = 0.1$.}
    \label{fig:syn_k_recall}
  \end{subfigure}

  \vspace{0.6em}

  % ---- Row 2 ----
  \begin{subfigure}[t]{0.31\textwidth}
    \centering
    \includegraphics[width=\linewidth]{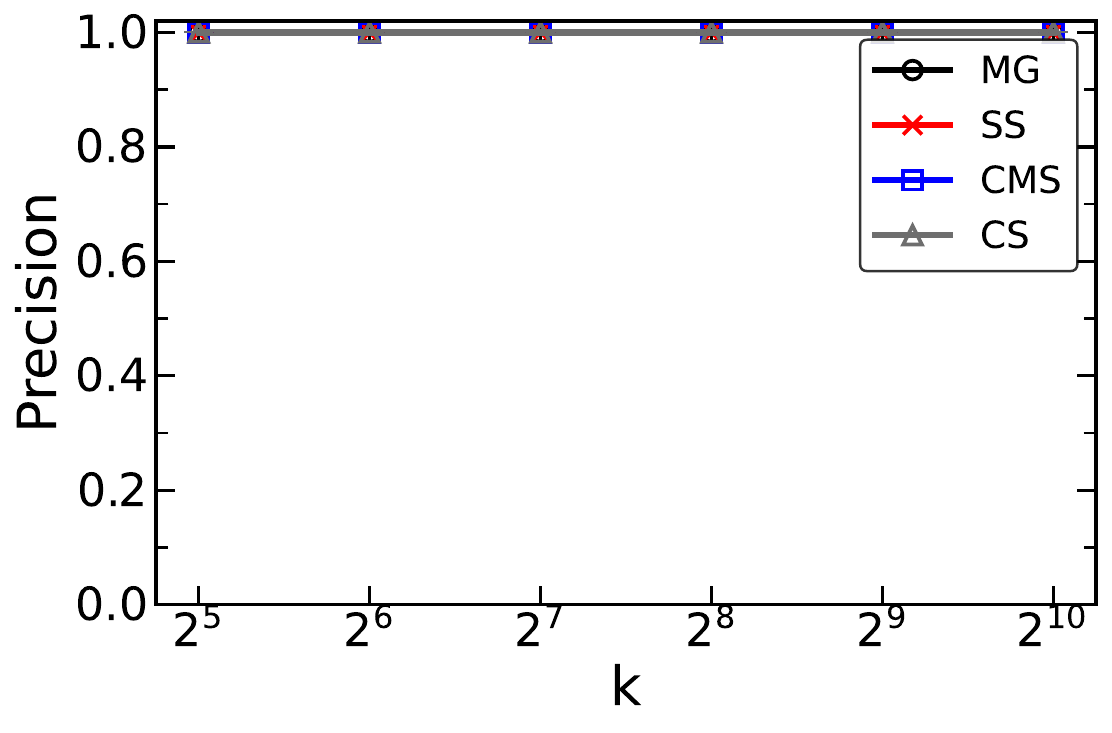}
    \caption{Precision vs $k$. Skew is $1.1$ and $\varepsilon = 0.1$.}
    \label{fig:syn_k_precision}
  \end{subfigure}\hfill
  \begin{subfigure}[t]{0.31\textwidth}
    \centering
    \includegraphics[width=\linewidth]{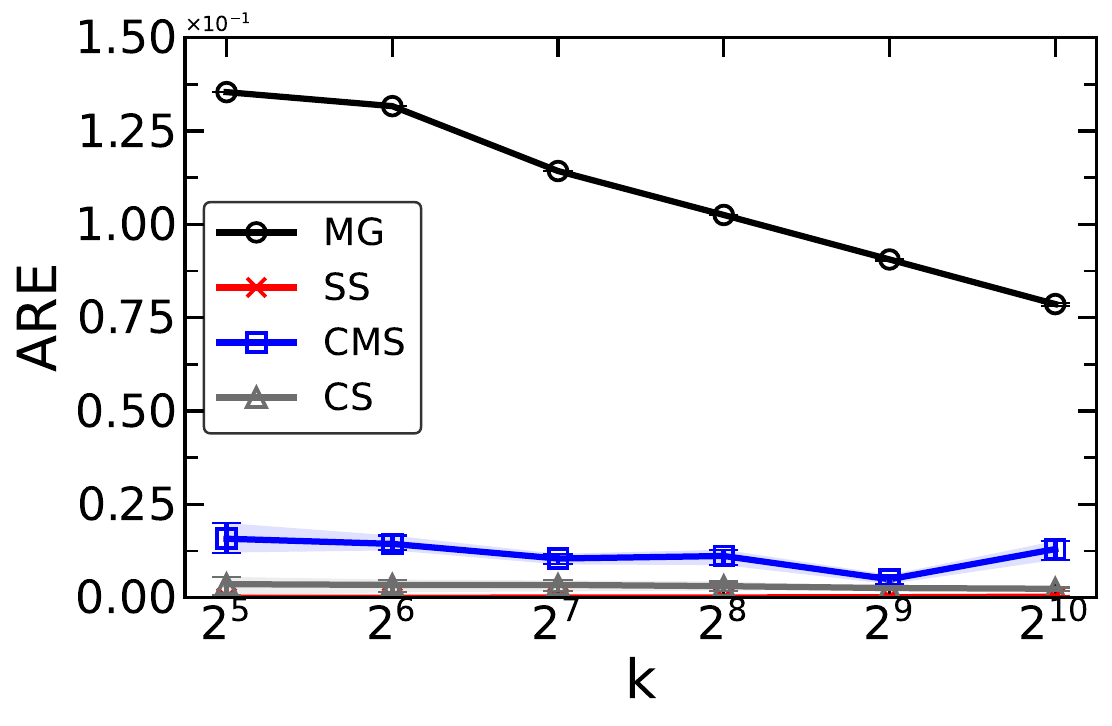}
    \caption{ARE vs $k$. Skew is $1.1$ and $\varepsilon = 0.1$.}
    \label{fig:syn_k_ARE}
  \end{subfigure}\hfill
  \begin{subfigure}[t]{0.31\textwidth}
    \centering
    \includegraphics[width=\linewidth]{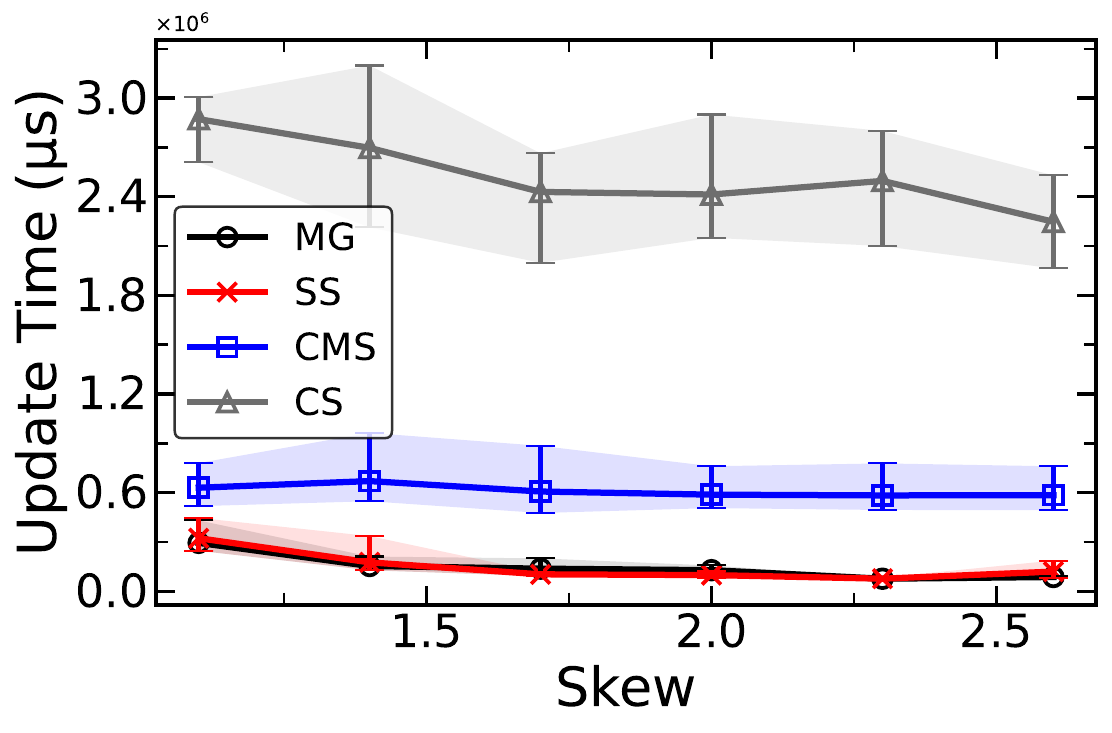}
    \caption{Time vs skew. Parameters: $k = 128$, $\varepsilon = 0.1$.}
    \label{fig:syn_skew_update}
  \end{subfigure}

  \caption{Experiments on synthetic data}
  \label{fig:syn_k}
\end{figure*}

\begin{proof}
At time $t$, every label in $\mathcal{T}\pr$ had estimate at least the $\tilde{k}$-th order statistic.
The maximum value of the $\tilde{k}$-th order statistic is bounded above by ${t}/{\tilde{k}}+\gamma_2(t)$: at most $\tilde{k}$ labels can have true frequency exceeding $t/\tilde{k}$, and the sketch's lower error bound inflates estimates by no more than $\gamma_1(t)$. 
Meanwhile, $y$'s estimate can be below its exact frequency by up to $\gamma_2(t)$. 
Combining, and using the monotonicity of $\gamma_1,\gamma_2$, yields the stated bound on $\tilde{f}_T(y)$ at release time. 
\end{proof}

As $\gamma_1(T)$ provides an upper bound on the error of each estimate in Algorithm~\ref{alg:EEHH},
by Lemma~\ref{lem:suppress}, an item $x \in \mathcal{T}\Delta \mathcal{T}\pr$ in the symmetric difference has frequency estimate at most
\[
  \hat{f}_T(x) \leq f_T(x) + \gamma_1(T) \leq \frac{t}{\tilde{k}} + \gamma_2(T) + 2\gamma_1(T) \leq \tau.  
\]
Therefore, unstable items are safely suppressed by the single threshold.
This allows us to establish that, conditioned on the Error Envelope, Algorithm~\ref{alg:EEHH} is as private as the underlying sketch.

\begin{restatable}{theorem}{privfreqoracle}
\label{thm:privacy-ls}
If the sketch $\mathcal{D}$ is $\varepsilon$-DP under single observation and incurs the following error bound for any $x \in X$
\[
\gamma_1(t)\ \ge\ \tilde{f}_t(x)-f_t(x)\ \ge\ \gamma_2(t)
\]
with probability at least $1-\delta/(2(T + \tilde{k}))$,
then Algorithm~\ref{alg:EEHH} is $(\varepsilon,\delta)$-differentially private.
\end{restatable}
\noindent
The proof is relegated to Appendix~\ref{app:proof_2} due to its similarities to Theorem~\ref{thm:priv_ss}.
The concrete envelopes and tail probabilities depend on $\mathcal{D}$, but the wrapper is agnostic to the internal design.
We give examples of Error Envelopes for linear sketches in Appendix~\ref{app:error_envelope}.

\section{Experiments}
\label{sec:experiments}

In this paper, we have proposed novel $\diffpriv$ heavy hitter algorithms tailored to the single observation streaming model.
Our central contribution is a $\diffpriv$ adaptation of $\spacesaving$, a state-of-the-art heavy hitters algorithm in non-private streams.
We posit that the strong empirical performance of $\spacesaving$ in the non-private setting carries over to the single-observation $\diffpriv$ regime.
 
To substantiate this claim, we implement a suite of $\diffpriv$ heavy hitter algorithms and conduct an empirical evaluation, analyzing the trade-offs between privacy, utility, and computational performance.

\subsection{Implemented Algorithms}
The following algorithms are implemented and evaluated for heavy hitter detection:
\begin{itemize}
    \item Differentially private $\mg$ \cite{lebeda2023better}.
    As prior work, it serves as the key baseline.
    \item Differentially private $\spacesaving$. 
    \item The Error Envelope Heavy Hitters implemented with both the $\cms$ and $\cs$.
    The internal Error Envelopes are discussed in Appendix~\ref{app:error_envelope}.
    Note that the $\cs$ additionally requires estimation of $F_2$, which can be done using the rows of the sketch.
\end{itemize}
The goal of each algorithm is to admit items with frequency greater than $T/k$, for stream length $T$.
The instantiations of differentially private $\spacesaving$, $\cms$, and $\cs$ depend on the additional oversampling parameter $\tilde{k}$, which, through a trade-off between memory and error,  allows isolated elements to be suppressed without impeding recall.
For linear sketches, to accommodate the Error Envelope, we set $\tilde{k} = 4k$.
For the counter-based algorithms, we set $\tilde{k}=2k$.
This results in a slight modification to $\mg$, where we increase the number of counters to $\tilde{k}$ and (similar to $\spacesaving$) admit items that exceed the maximum of $T/k$ and the suppression threshold.

\begin{figure*}[t]
  \centering

  % ---- Row 3 ----
  \begin{subfigure}[t]{0.31\textwidth}
    \centering
    \includegraphics[width=\linewidth]{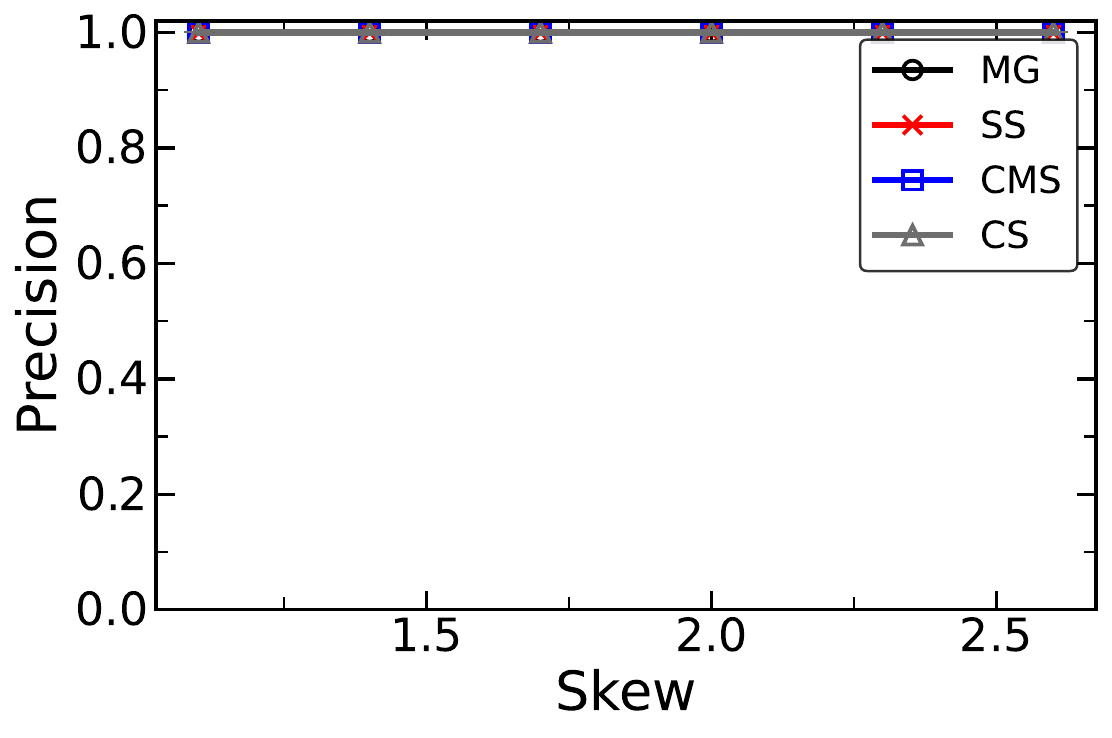}
    \caption{Precision vs Skew. Parameters: $k = 128$, $\varepsilon = 0.1$.}
    \label{fig:syn_skew_precision}
  \end{subfigure}\hfill
  \begin{subfigure}[t]{0.31\textwidth}
    \centering
    \includegraphics[width=\linewidth]{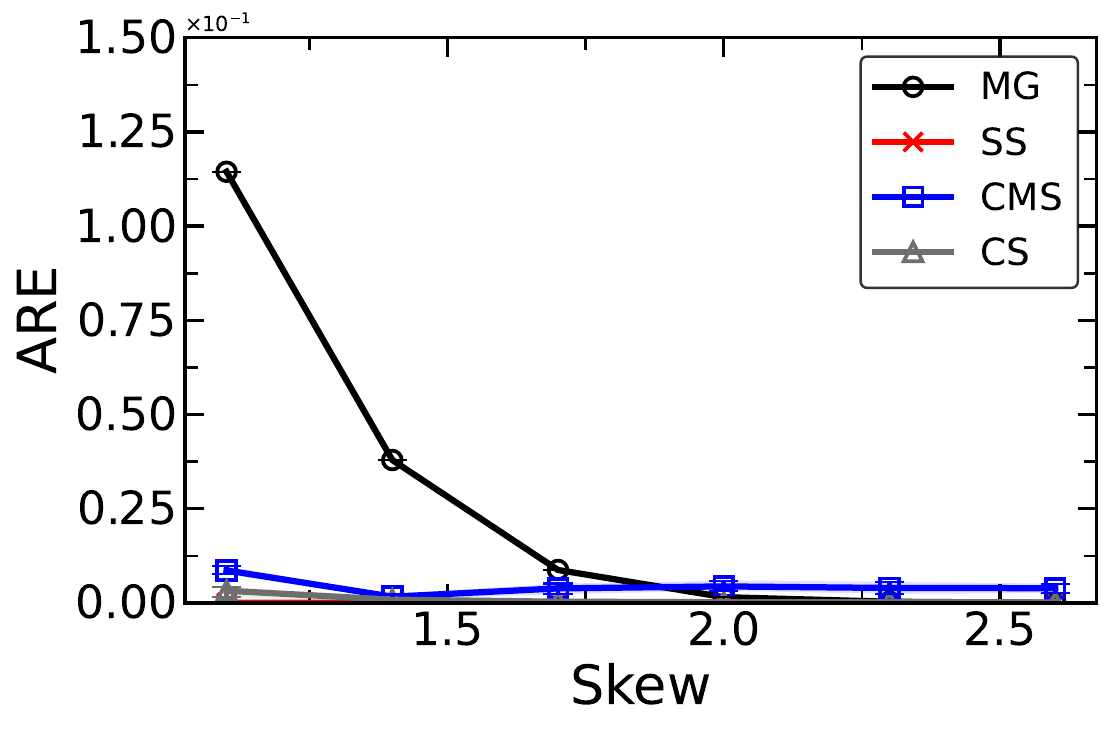}
    \caption{ARE vs Skew. Parameters: $k = 128$, $\varepsilon = 0.1$.}
    \label{fig:syn_skew_ARE}
  \end{subfigure}\hfill
  \begin{subfigure}[t]{0.31\textwidth}
    \centering
    \includegraphics[width=\linewidth]{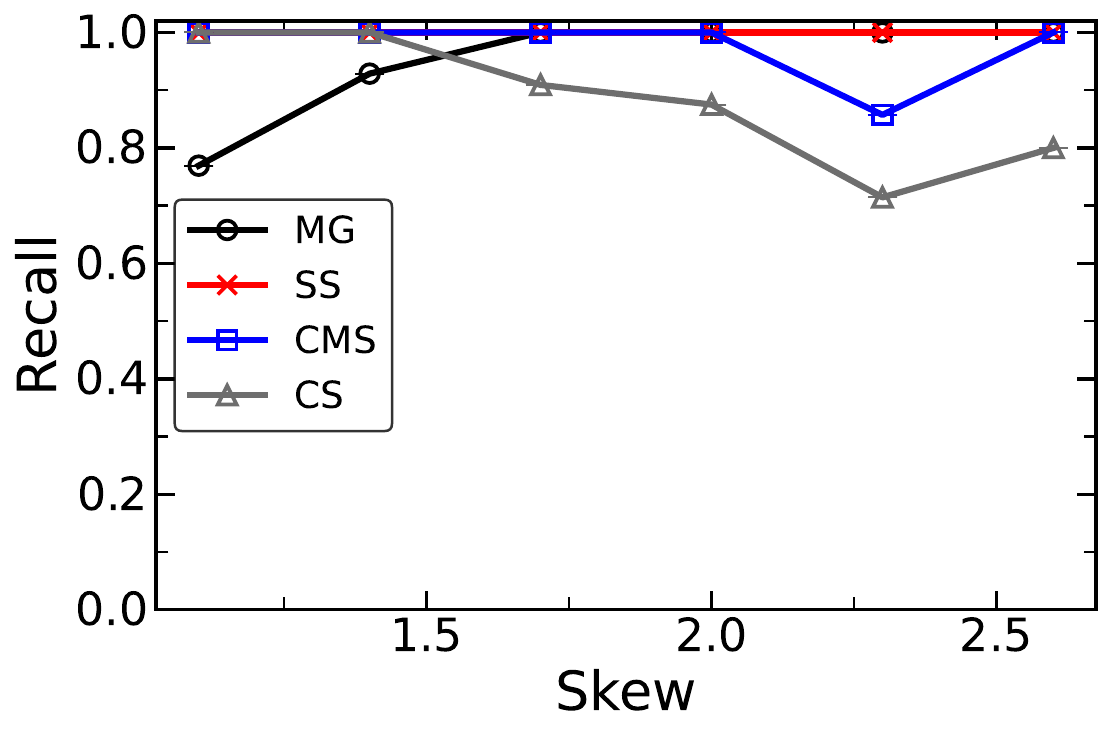}
    \caption{Recall vs Skew. Parameters: $k = 128$, $\varepsilon = 0.1$.}
    \label{fig:syn_skew_recall}
  \end{subfigure}

  \vspace{0.6em}

  % ---- Row 4 ----
  \begin{subfigure}[t]{0.31\textwidth}
    \centering
    \includegraphics[width=\linewidth]{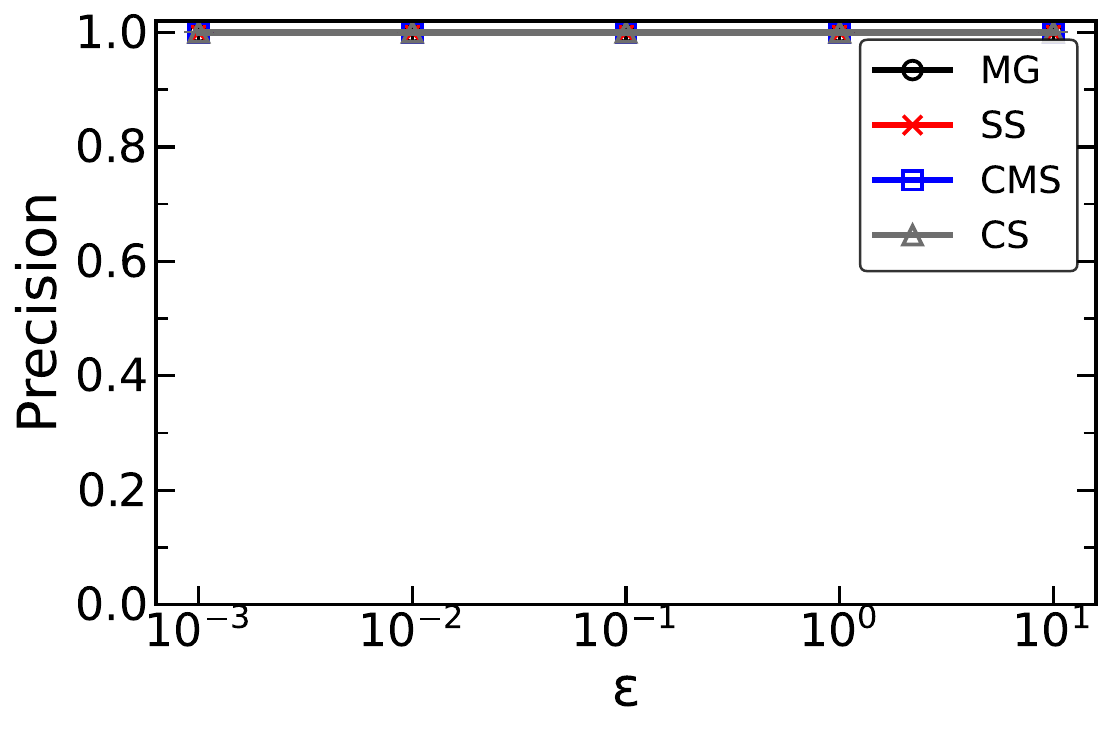}
    \caption{Precision vs $\varepsilon$. Skew is $1.1$ and $k = 128$.}
    \label{fig:syn_eps_precision}
  \end{subfigure}\hfill
  \begin{subfigure}[t]{0.31\textwidth}
    \centering
    \includegraphics[width=\linewidth]{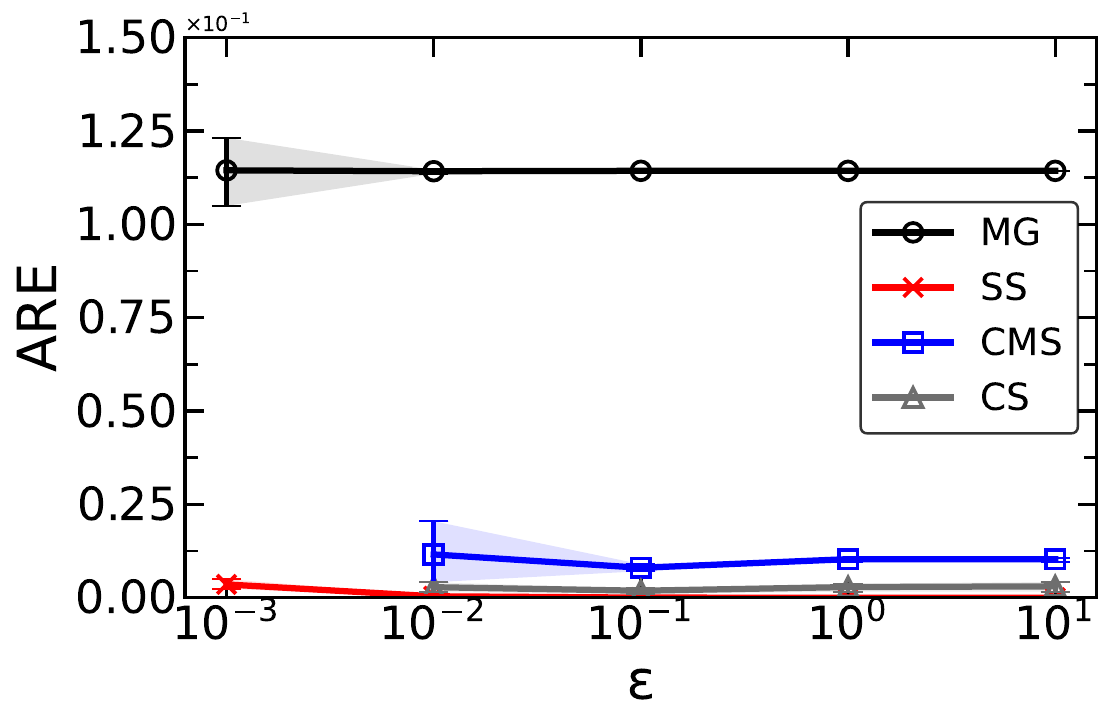}
    \caption{ARE vs $\varepsilon$. Skew is $1.1$ and $k = 128$.}
    \label{fig:syn_eps_ARE}
  \end{subfigure}\hfill
  \begin{subfigure}[t]{0.31\textwidth}
    \centering
    \includegraphics[width=\linewidth]{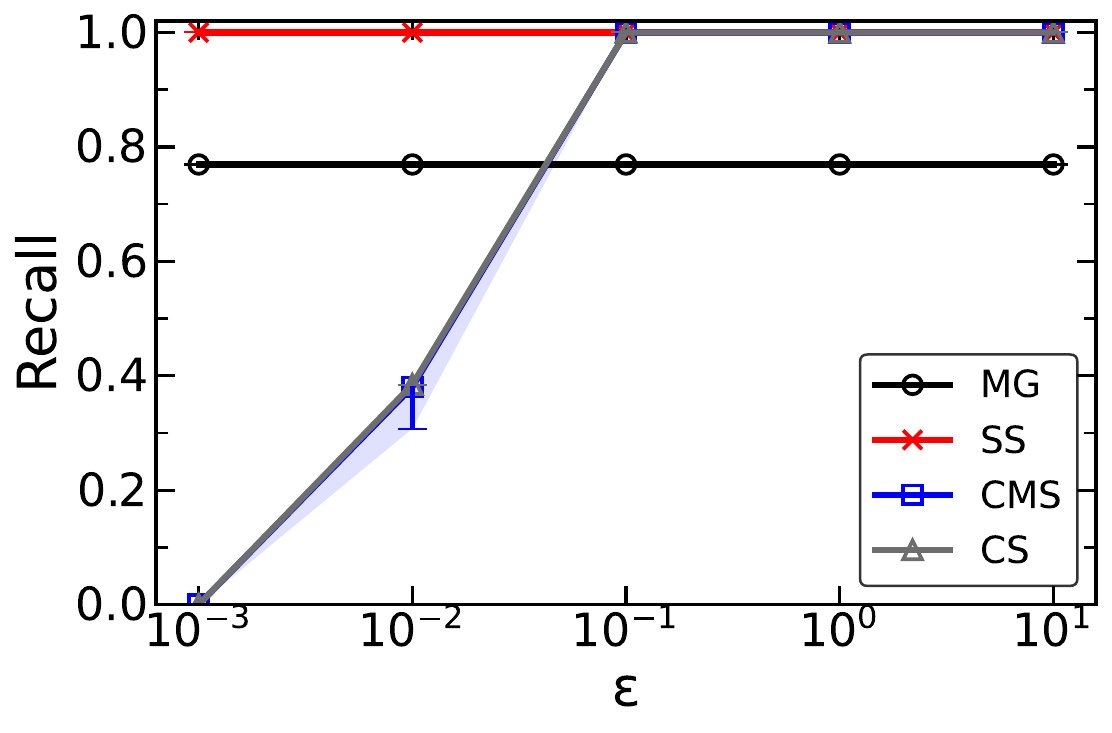}
    \caption{Recall vs $\varepsilon$. Skew is $1.1$ and $k = 128$.}
    \label{fig:syn_eps_recall}
  \end{subfigure}

  \caption{Experiments on Syntetic Data}
  \label{fig:2x3_part2}
\end{figure*}

\subsection{Datasets and Metrics}

We ran experiments using both real network traffic and generated synethetic
data. 
The network data set was drawn from anonymized traffic traces from the CAIDA 2019 dataset\footnote{\url{https://www.caida.org/data/passive/passive\_2019\_dataset.xml}}.
The dataset consists of approximately $26.6$ million total source IP addresses and around $312{,}000$ distinct source IPs.
We generated synthetic data from a skewed distribution (Zipf), varying the skew from 1.1 to 2.7 (in order to obtain meaningful distributions that produce at
least one heavy hitter per run). 

For both synthetic and real data we varied two parameters: the heavy hitter threshold $k$  and the privacy parameter $\varepsilon$.
The privacy failure probability $\delta$ was set to $0.001$ for all experiments.

Following \cite{cormode2008finding}, we evaluate the algorithms along several dimensions of efficiency and accuracy:
\begin{itemize}
\item {Update throughput}, reported as the average time per update in microseconds.
\item {Memory footprint}, measured in bytes.
\item {Recall}, defined as the fraction of true heavy hitters (as identified by an exact baseline) that are successfully reported.
\item {Precision}, defined as the fraction of reported items that are true heavy hitters, capturing the rate of false positives.
\item {Average relative error} of the estimated frequencies for reported items.
\end{itemize}
Each experimental configuration is repeated 20 times.
Algorithms are evaluated independently to avoid interference from shared caching effects.
All plots report mean values, with error bars indicating the 5th and 95th percentiles.

All implementations are written in C++ and is publicly available\footnote{\url{https://github.com/rayneholland/DPHH}}. 
Experiments were conducted on a Dell Latitude 7430 running Windows 11, equipped with an Intel Core i7 processor and 32GB of RAM. 

\begin{figure*}[t]
  \centering

  % ---- Row 1 ----
  \begin{subfigure}[t]{0.31\textwidth}
    \centering
    \includegraphics[width=\linewidth]{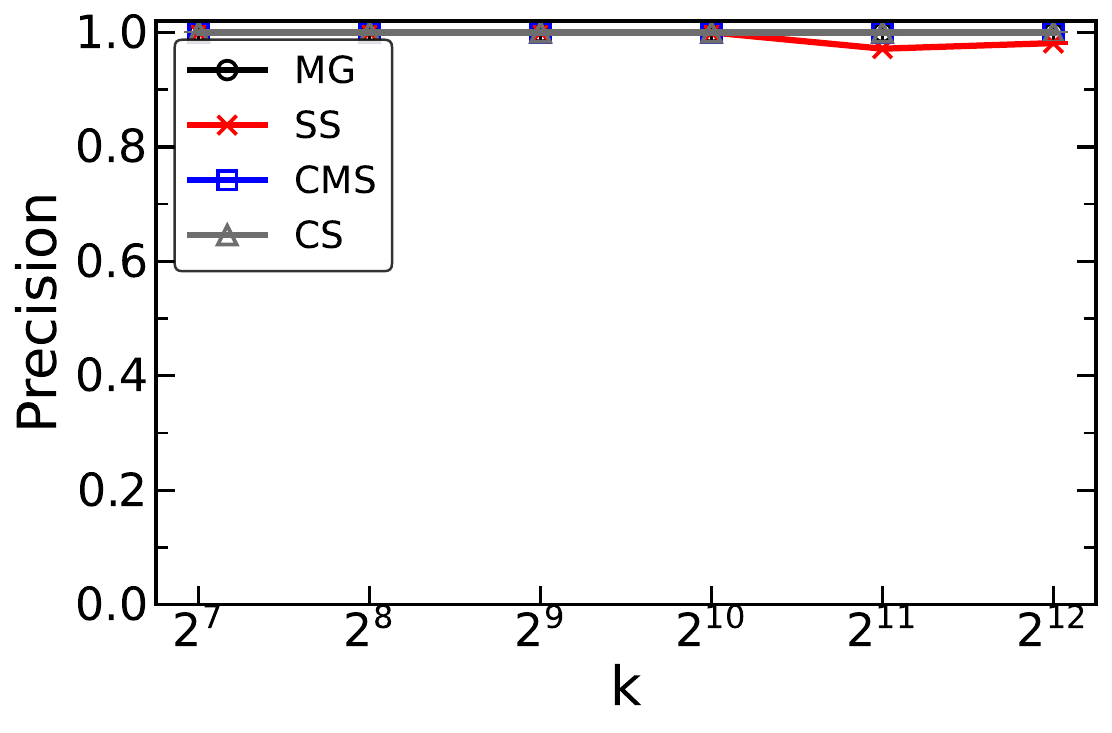}
    \caption{Precision vs $k$. Privacy $\varepsilon = 0.1$.}
    \label{fig:caida_k_precision}
  \end{subfigure}\hfill
  \begin{subfigure}[t]{0.31\textwidth}
    \centering
    \includegraphics[width=\linewidth]{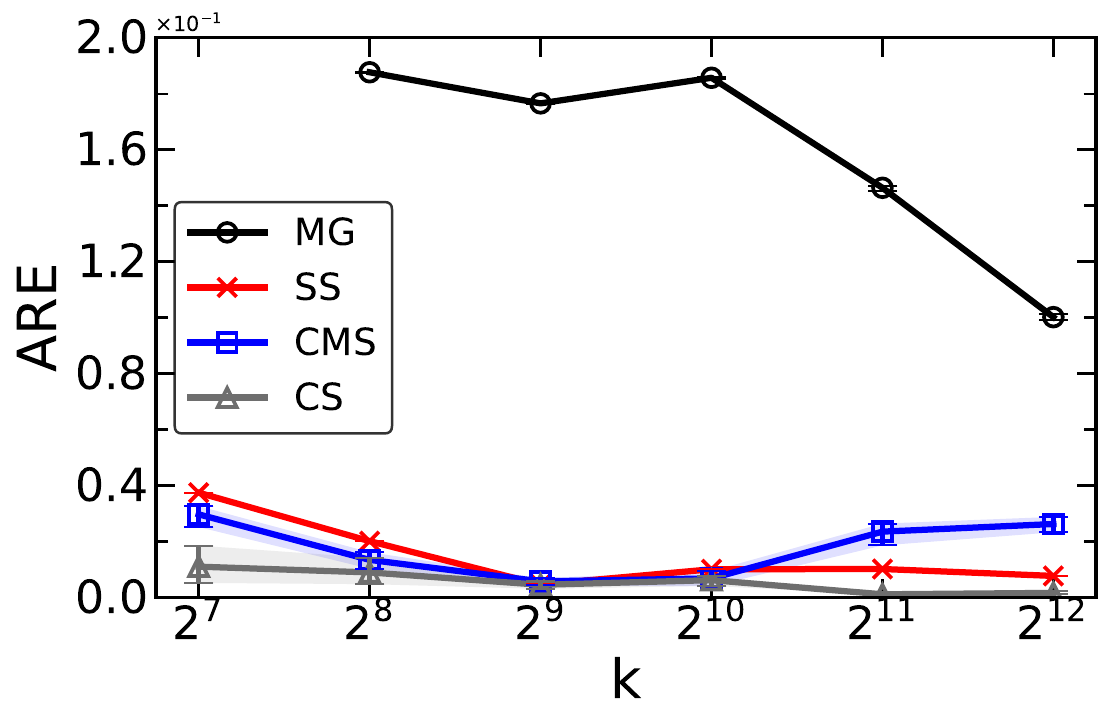}
    \caption{ARE vs $k$. Privacy $\varepsilon = 0.1$.}
    \label{fig:caida_k_ARE}
  \end{subfigure}\hfill
  \begin{subfigure}[t]{0.31\textwidth}
    \centering
    \includegraphics[width=\linewidth]{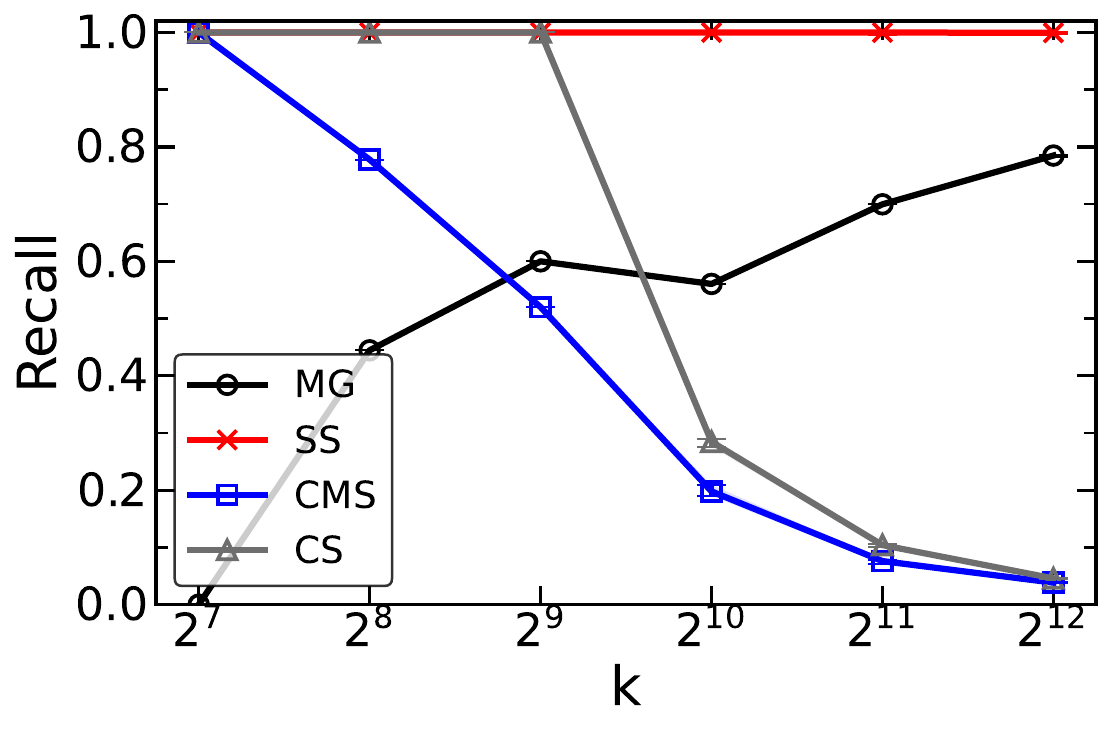}
    \caption{Recall vs $k$. Privacy $\varepsilon = 0.1$.}
    \label{fig:caida_k_recall}
  \end{subfigure}

  \vspace{0.6em}

  % ---- Row 2 ----
  \begin{subfigure}[t]{0.31\textwidth}
    \centering
    \includegraphics[width=\linewidth]{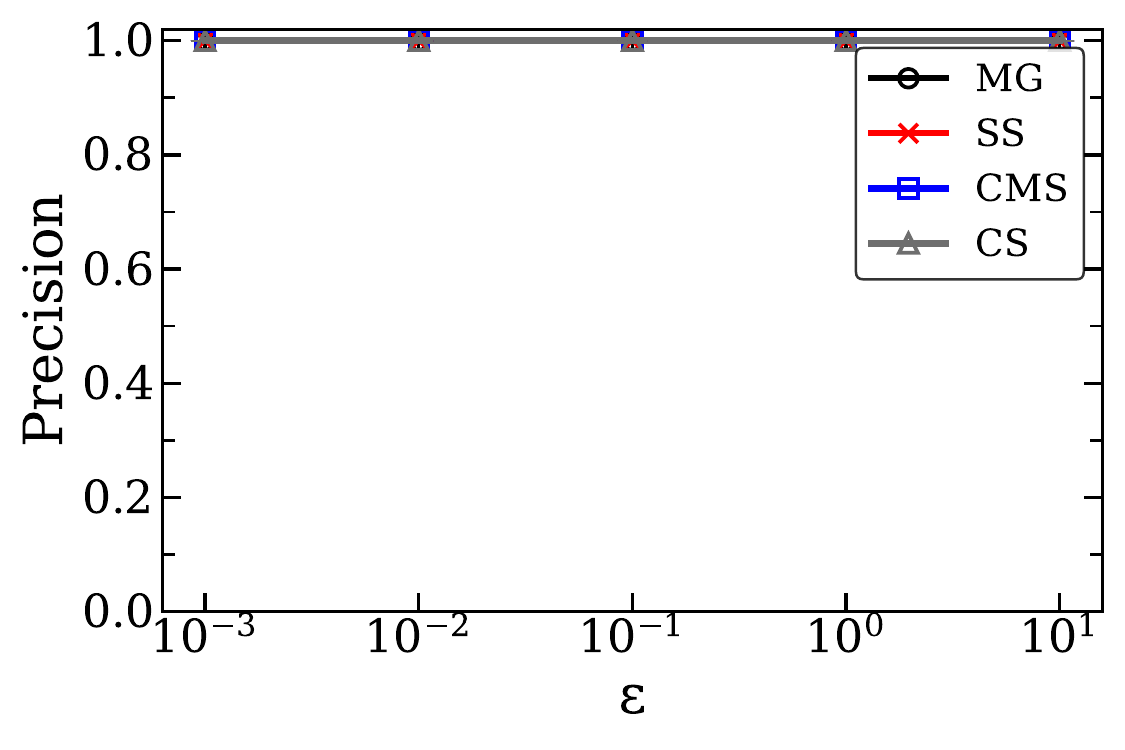}
    \caption{Precision vs $\varepsilon$. Parameter $k = 256$.}
    \label{fig:caida_eps_precision}
  \end{subfigure}\hfill
  \begin{subfigure}[t]{0.31\textwidth}
    \centering
    \includegraphics[width=\linewidth]{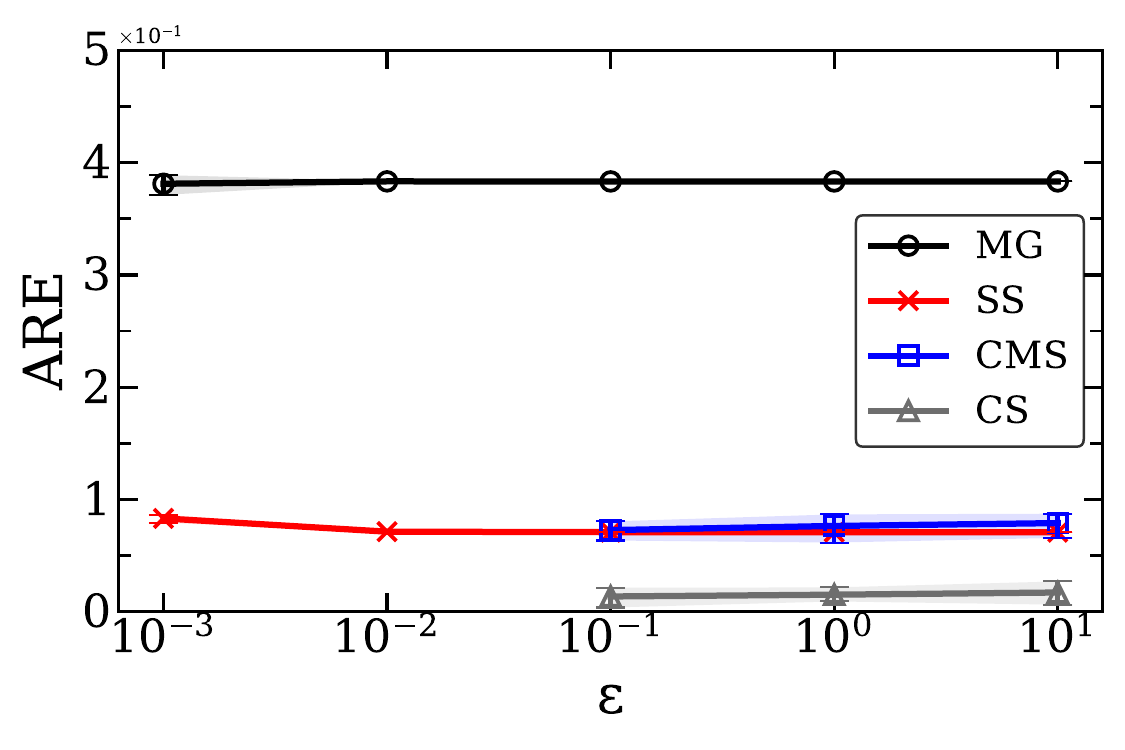}
    \caption{ARE vs $\varepsilon$. Parameter $k = 256$.}
    \label{fig:caida_eps_ARE}
  \end{subfigure}\hfill
  \begin{subfigure}[t]{0.31\textwidth}
    \centering
    \includegraphics[width=\linewidth]{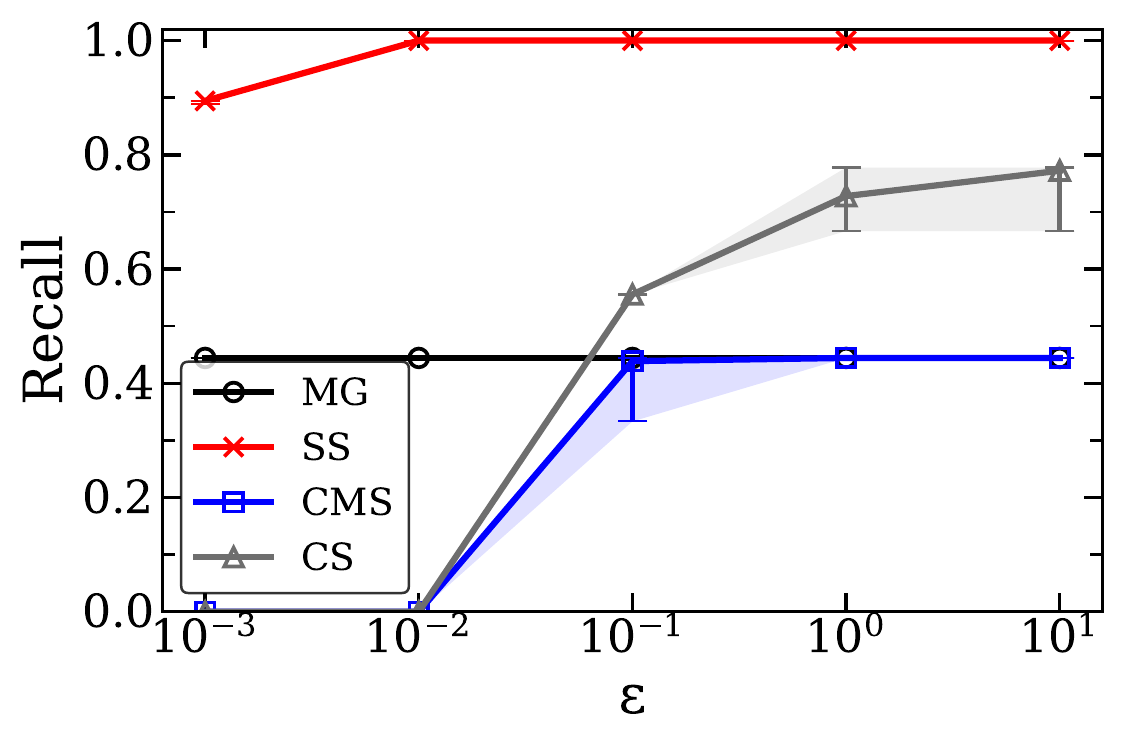}
    \caption{Recall vs $\varepsilon$. Parameter $k = 256$.}
    \label{fig:caida_eps_recall}
  \end{subfigure}

  \vspace{0.6em}

  % ---- Row 3 ----
  \begin{subfigure}[t]{0.31\textwidth}
    \centering
    \includegraphics[width=\linewidth]{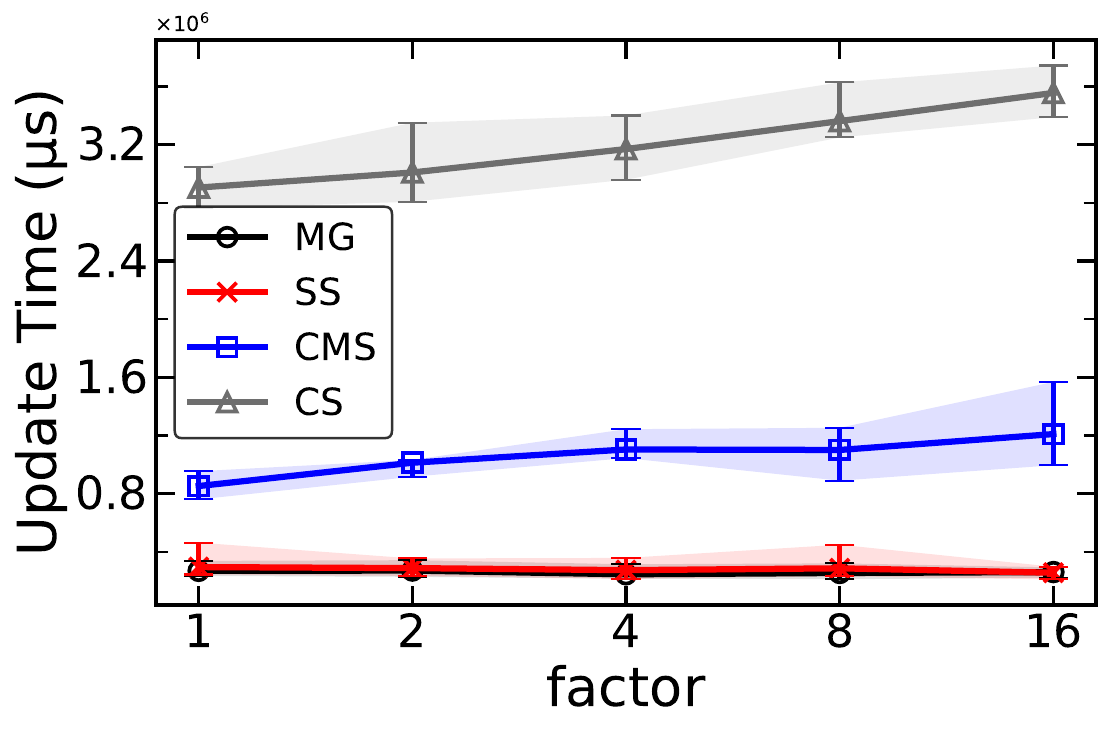}
    \caption{Time vs $\tilde{k} = c\cdot k$. Parameter $k = 1024$.}
    \label{fig:caida_tilde_k_update}
  \end{subfigure}\hfill
  \begin{subfigure}[t]{0.31\textwidth}
    \centering
    \includegraphics[width=\linewidth]{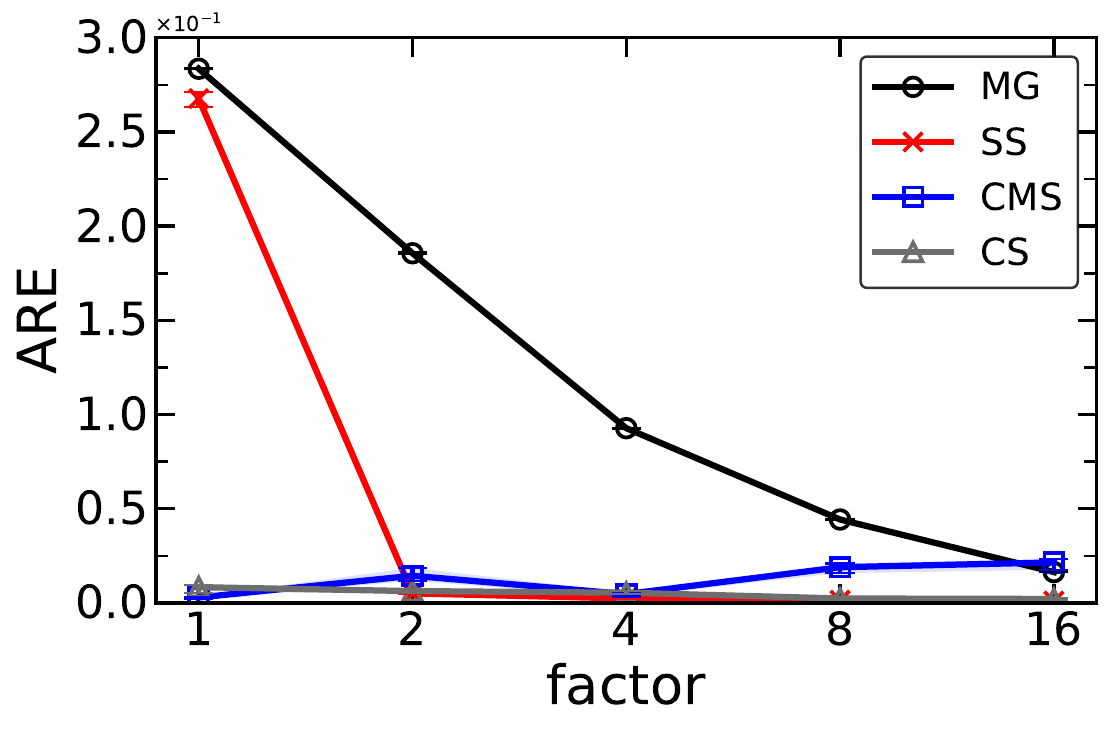}
    \caption{ARE vs $\tilde{k} = c\cdot k$. Parameter $k = 1024$.}
    \label{fig:caida_tilde_k_ARE}
  \end{subfigure}\hfill
  \begin{subfigure}[t]{0.31\textwidth}
    \centering
    \includegraphics[width=\linewidth]{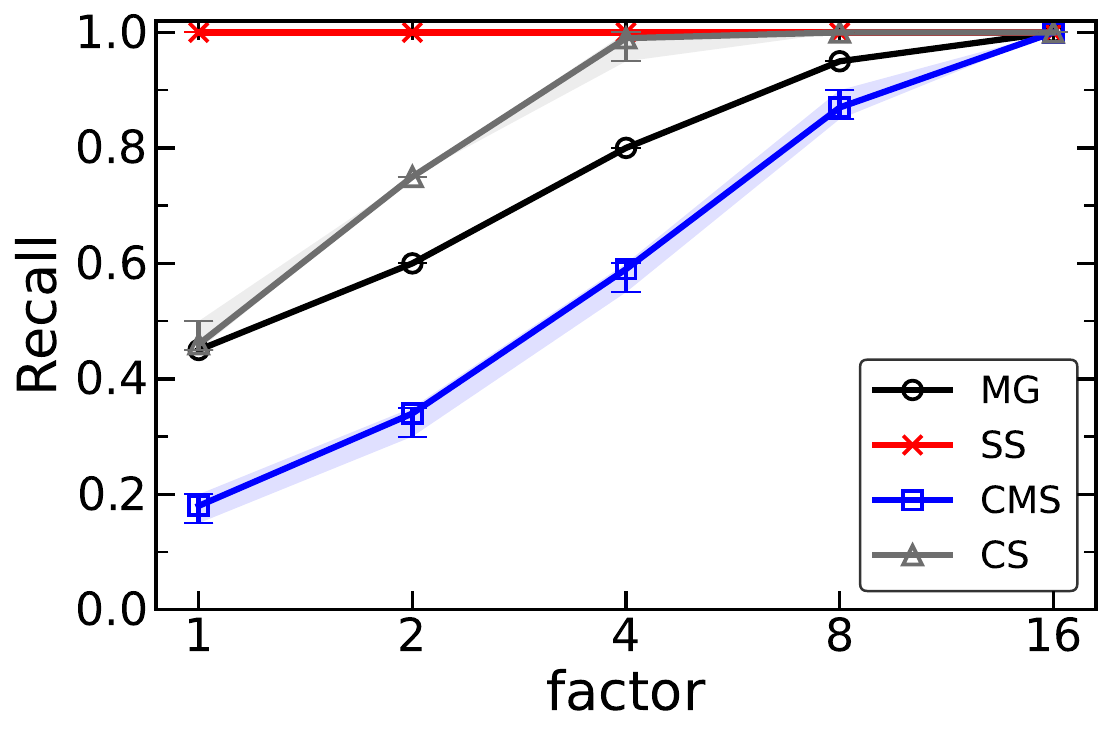}
    \caption{Recall vs $\tilde{k} = c\cdot k$. Parameter $k = 1024$.}
    \label{fig:caida_tilde_k_recall}
  \end{subfigure}

  \caption{Experiments on CAIDA data.}
  \label{fig:three_by_three}
\end{figure*}

\subsection{Synthetic Data}

\paragraph{Varying $k$.}
Figure~\ref{fig:syn_k_bytes} shows that, at a fixed $k$, the linear sketches consume {significantly} more memory than the counter-based summaries, roughly a factor of six.
This disparity stems from the large depth ($d\!>\!20$) required needed to enforce the error envelope.
The counter-based algorithms require at most 240KB of memory, underscoring their compactness.
Using a heap-based implementation the memory allocation of counter-based algorithm can be increased further, albeit at the cost of a slightly increased update time. 

Update throughput (Figure~\ref{fig:syn_k_update}) favors the counter-based methods by a wide margin.
The sketches show only a weak dependence on $k$ due to heap-based candidate tracking. 
Recall against the frequency threshold (Figure~\ref{fig:syn_k_recall}) is perfect for $\spacesaving$, improves steadily for $\mg$ as the number of registers grows, and degrades for the linear sketches as the envelope is tightens. 
This degradation arises because linear sketches require substantially more injected noise; as $T/k$ approaches the bound imposed by the privacy-induced noise (for larger $k$), recall correspondingly declines.
Precision is uniformly $1.0$ across all four algorithms (Figure~\ref{fig:syn_k_precision}) under our thresholding rule. 
For ARE, $\mg$ exhibits the largest error, $\spacesaving$ performs best, and $\cs$ has a consistent edge over $\cms$ (Figure~\ref{fig:syn_k_ARE}).

\paragraph{Varying skew $\phi$.}
As skew increases, throughput improves modestly for all methods (Figure~\ref{fig:syn_skew_update}), reflecting fewer candidate replacements when the head of the distribution dominates. 
Precision remains perfect across all four algorithms (Figure~\ref{fig:syn_skew_precision}). 
ARE decreases with lower skew (Figure~\ref{fig:syn_skew_ARE}), as is expected in frequency estimation.
As before, $\spacesaving$ and $\cs$ achieve the lowest error. 
Recall (Figure~\ref{fig:syn_skew_recall}) is perfect for $\spacesaving$, improves for $\mg$ as skew decreases, and is lowest for $\cs$ due to its larger envelope driven by the need to estimate and bound $F_2$.

\paragraph{Varying privacy $\varepsilon$.}
Precision is unaffected by privacy and remains perfect (Figure~\ref{fig:syn_eps_precision}). 
Privacy impacts ARE primarily, in the form of higher estimation variance, at tight budgets: for very small $\varepsilon$, error is influenced by the noise term (Figure~\ref{fig:syn_eps_ARE}).
The $\cms$ is most sensitive because minimum-of-rows estimation biases toward negative fluctuations.
This is also compounded by requiring a large depth to secure the error envelope.
%In this instance the error is reduced, however, there are settings where the error is increased significantly at lower $\varepsilon$ values due to this bias.
In Contrast, median-of-rows in $\cs$ composes with noise error better and yields a lower ARE. 
Recall (Figure~\ref{fig:syn_eps_recall}) is essentially unchanged for the counter-based methods (with a minor dip for $\mg$ only at very small $\varepsilon$), while the sketches benefit as $\varepsilon$ increases.
This is because the envelope shrinks and heavy hitters are admitted more reliably.

\paragraph{Main takeaway.}
Across synthetic streams, $\spacesaving$ delivers the best overall utility (perfect recall and precision, lowest ARE) and the highest throughput, while using substantially less memory than the linear sketches. 
$\cs$ consistently outperforms $\cms$ on error, particularly under privacy noise. 
%Increasing $\tilde{k}$ improves utility for all data structures, and higher skew modestly reduces update time. 
These trends collectively indicate that our $\diffpriv$ instantiations preserve the empirical strengths and performance of the underlying non-private algorithms, with $\spacesaving$ retaining its advantage across all metrics.

\subsection{Real Data}

\paragraph{Varying $k$.}
Precision is near-perfect across all methods (Figure~\ref{fig:caida_k_precision}), with $\spacesaving$ admitting a small number of false positives as the threshold is relaxed. 
ARE generally declines with larger $k$ (Figure~\ref{fig:caida_k_ARE}).
$\cms$ exhibits a modest increase in error as more lower-frequency items are admitted as heavy, while $\cs$ and $\spacesaving$ maintain error below $0.04$ throughout. 
Recall trends diverge (Figure~\ref{fig:caida_k_recall}) between algorithms.
Notably, $\mg$ benefits from larger tables, whereas both linear sketches degrade as $k$ increases, a consequence of the privacy-induced noise dominating the frequency threshold $T/k$.

\paragraph{Varying $\varepsilon$.}
Precision is unaffected and remains perfect for all data structures (Figure~\ref{fig:caida_eps_precision}). 
ARE is essentially insensitive to $\varepsilon$ (Figure~\ref{fig:caida_eps_ARE}), as estimation error from the underlying non-private summaries dominates the total error in the $\diffpriv$ variants.
Recall is impacted only at very small privacy budgets (Figure~\ref{fig:caida_eps_recall}), where the suppression thresholds for linear sketches include an $\varepsilon^{-1}$ term.

\paragraph{Varying $\tilde{k}$.}
Increasing the oversampling factor $c$ in $\tilde{k}=c\cdot k$ leaves update time unchanged for the counter-based summaries (Figure~\ref{fig:caida_tilde_k_update}), consistent with constant-time updates. 
Linear sketches, which maintain heaps for candidate tracking, show a dependence on $\tilde{k}$. 
By design, utility improves uniformly with larger $\tilde{k}$.
ARE decreases as additional counters reduce collision and estimation error (Figure~\ref{fig:caida_tilde_k_ARE}).
Note that $\spacesaving$ obtains higher error at Factor 1 as it is the only algorithm with high recall and therefore gets tested against lower frequency heavy items.
Recall rises because the expanded capacity tightens error envelopes and enables finer thresholding without violating privacy (Figure~\ref{fig:caida_tilde_k_recall}).

\paragraph{Main takeaway.}
$\spacesaving$ retains its empirical edge, delivering high precision and recall across configurations.
Linear sketches remain competitive on error but can suffer recall penalties under tight envelopes. 
Notably, the oversampling parameter $\tilde{k}$ is an effective knob for improving recall, at the cost of additional memory.

\section{Conclusion}

We introduced new methods for privatizing heavy hitter algorithms in the data-stream model of computation. 
This includes the first $\diffpriv$ variant of $\spacesaving$ and a generic wrapper that converts any $\diffpriv$ frequency oracle into a label-private heavy-hitter mechanism.
Both constructions combine calibrated noise with a single threshold tied to the error envelope, preserving the efficiency of their non-private counterparts.

Our experiments confirm that $\spacesaving$ continues to dominate under differential privacy, achieving near-perfect precision and recall with minimal overhead.
$\mg$ and $\spacesaving$ largely match the empirical performance of their non-private versions. 
Linear sketches remain close to their non-private behavior but incur a small penalty from deeper tables required to guarantee envelopes.
Linear sketches nonetheless offer advantages such as deletions. 
A promising direction for future work is to further tighten sketch error envelopes.

% ---------- Bibliography ----------
\bibliographystyle{plain}
\bibliography{references}

\appendix

\section{Related Work}
\label{app:related_work}

Heavy hitter detection has been studied extensively in the streaming model, with a broad taxonomy spanning \emph{counter-based} ($\mg$ \cite{misra1982finding}, $\spacesaving$ \cite{mitzenmacher2012hierarchical}), \emph{sketch-based} ($\cms$ \cite{cormode2005improved}, $\cs$ \cite{charikar2002finding}), and \emph{quantile}-based approaches \cite{greenwald2001space, shrivastava2004medians}. 
The most comprehensive empirical study is due to Cormode and Hadjieleftheriou, who implemented the principal families under uniform conditions and reported accuracy–space–throughput trade-offs \cite{cormode2008finding}. 
Their results show that, for insertion-only streams, $\spacesaving$ achieves highest accuracy with small memory and high throughput, while sketch-based summaries offer deletion support at the cost of more memory and less throughout. 
We extend this experimental lens to the privacy setting, measuring performance of differentially private counter- and sketch-based algorithms, and observe that the original conclusions still hold under privatization.

A number of works have explored differential privacy in the single-observation model \cite{chan2012differentially, holland2025private, lebeda2023better, pagh2022improved, zhao2022differentially}.
Early attempts to privatize $\mg$ used global sensitivity arguments, which implied noise scaling with the number of counters \cite{chan2012differentially}. 
Lebeda and Tětek resolved this by exploiting the fine-grained structure of neighboring $\mg$ sketches \cite{lebeda2023better}.
Their analysis reduces the effective sensitivity to \(O(1)\) (per coordinate) by leveraging shared offsets between neighboring counters.
Their approach yields \((\varepsilon,\delta)\)-DP with worst-case error matching the best non-streaming private histograms up to constants. 
Despite strong practical performance, $\spacesaving$ (which typically dominates $\mg$ empirically) has received little attention in the privacy literature. 

Linear sketches provide frequency vector summaries with rigorous probabilistic error guarantees, and have long been used for frequency estimation and frequent item detection. 
Recent work shows how to {privatize} these sketches in the single-observation $\diffpriv$ model. 
Pagh and Thorup analyzed \emph{private $\cs$}, giving improved utility bounds, showing that the error of the query method composes  nicely with Gaussian noise \cite{pagh2022improved}.
Zhao et al. developed efficient private linear-sketch implementations and demonstrated strong accuracy-space trade-offs in practice \cite{zhao2022differentially}. 
These mechanisms output $\diffpriv$ frequency estimates for \emph{queried} items. 
Producing heavy hitters, however, generally requires scanning or iterating over the universe to identify large counts.
This is an operational hurdle in large domains.

A common workaround uses a $\diffpriv$ {$\cs$ as a frequency oracle} inside a heavy-hitter pipeline \cite{bassily2017practical}. 
While effective, and improve over brute-force iteration, current instantiations still rely on substantial post-processing requirements to prune candidates and locate frequent items. 
These post-processing steps require a central server.
Our work bypasses this requirement and allows 
More broadly, heavy hitter detection has been given a lot of attention in the local DP model \cite{bun2019heavy,erlingsson2014rappor,qin2016heavy,wu2022asymptotically}, where users send one randomized report and the server reconstructs frequent items.

\section{Error Envelopes for Linear Sketches}
\label{app:error_envelope}

The non-private $\cms$ has the following error bound.
\begin{lemma}[Approximation error for $\cms$ \cite{cormode2005improved}]
    For $x \in \mathcal{U}$ and any $\beta, \eta \in (0,1)$, the estimation error of $\cms$ with  \( d = \lceil \log(1/\beta) \rceil \)  repetitions and table size $w = \lceil 2/\eta \rceil $ satisfies,
    \begin{align*}
        \prob[\hat{f}(x) - f(x) > \eta F_1] \leq \beta, 
    \end{align*}
    where $F_1 = \sum_{y\in \mathcal{U}} f(y)$.
    \label{lem:cms_approx}
\end{lemma}
Thus, by setting the width to $2\tilde{k}$, the error in $\cms$ is bounded by $t/\tilde{k}$, with probability $1-\beta$.
Moreover, a $\cms$ has sensitivity $2d$ on neighboring streams \cite{zhao2022differentially}.
Thus, by Lemma~\ref{eq:laplace-mech}, it is sufficient to add $\laplacedist(2d/\varepsilon)$ noise to each cell for $\varepsilon$-differential privacy.
By the Laplace tail bounds, with probability $1-\delta/2$, for $\psi = 2d/\varepsilon \ln (4\tilde{k}d/\delta)$, the additive error from a noise random variable is in the interval $[-\psi, \psi]$.
Therefore, by setting the depth to $d \lceil\ln(2(T+\tilde{k})/\delta)\rceil$, we can use $\gamma_1(t) = t/k + \psi$ and $\gamma_2(t) = -\psi$ for the Error Envelope.
A similar construction can be made for $\cs$ using the following result.

\begin{lemma}[Approximation error for $\cs$ \cite{charikar2002finding}]
    For \( x \in [n] \) and any \( \beta, \eta \in (0,1) \), the estimation error of Count Sketch with \( d = \lceil \ln(1/\beta) \rceil \) and table width \( \lceil w = 3/\eta^2 \rceil \) satisfies:
    \[
        \Pr\left[ |\hat{f}(x) - f(x)| > \eta\sqrt{F_2/w} \right] \leq \beta,
    \]
    where \( F_2 = \sum_{y \in [n]} f(y)^2 \) is the second moment of the frequency vector.
    \label{lem:cs_approx}
\end{lemma}

\section{Deferred Proofs Section~\ref{sec:soss}}
\label{app:proof}

\subsection{Lemma~\ref{lem:state_eos} Components}

\lemstwo*

\begin{proof}
If the process is in $\stwo$, we have to consider the cases $x_t \in \mathcal{T}_{t-1} \cap \mathcal{T}_{t-1}\pr$, $x_t \notin \mathcal{T}_{t-1} \cup \mathcal{T}_{t-1}\pr$, $x_t \in W_{t-1}$, or $x_t\in W_{t-1}\pr$.
For $x_t \in \mathcal{T}_{t-1}\cap \mathcal{T}_{t-1}\pr$, the sets of tracked items do not change and both $C_{t-1}[x_t]$ and $C_{t-1}\pr[x_t]$ are incremented to form $C_t$ and $C_t\pr$. 
As the counts of isolated elements do not change, the process stays in $\stwo$.

If $x_t \notin \mathcal{T} \cup \mathcal{T}\pr$, there are three possibilities, which are dependent on the identities of the items in the eviction registers.
\begin{itemize}
    \item[(i)] $\arg \max_{a\in S_t} \tau_{t-1}[a] = \arg \max_{a\in S_t\pr} \tau_{t-1}\pr[a]$. 
    This case follows the same logic as the `$x_t \in \mathcal{T}_{t-1}\cap \mathcal{T}_{t-1}\pr$' instance discussed above and we stay in $\stwo$.
    \item [(ii)] $z = \arg \max_{a\in S_t} \tau_{t-1}[a]$ and $z\pr = \arg \max_{a\in S_t\pr} \tau_{t-1}\pr[a]$.
    In this instance both $z$ and $z\pr$ are replaced by $x_t$ to form the sets $\mathcal{T}_t = \mathcal{T}_{t-1}$ (condition $\sone$(a)), with $C_t[x_t] = C_t\pr[x_t]+1$ (reflecting the count difference between $z$ and $z\pr$)  (condition $\sone$(b)).
    All other counts are left unchanged and we transition to $\sone$.
    \item [(iii)] $z \neq x^* = \arg \max_{a\in S_t} \tau_{t-1}[a]$ and $z\pr = \arg \max_{a\in S_t\pr} \tau_{t-1}\pr[a]$.
    The following tracked item replacements occur 
    \begin{align*}
        \mathcal{T}_{t} &\gets (\mathcal{T}_{t-1}\setminus \{x^*\}) \cup \{x_t\}, \\
        \mathcal{T}_{t}\pr &\gets (\mathcal{T}_{t-1}\pr\setminus \{z\pr\}) \cup \{x_t\}.
    \end{align*}
    This leads to $W_t = \{z\}$ and $W_t\pr = \{x^*\}$.
    As, $x^*$ in the eviction register of $\mathcal{T}_{t-1}$ implies that $x^*$ has the minimum count in $\mathcal{T}_{t-1}\pr\setminus \{z\pr\}$,
    %\[
    %    x^* = \arg \max_{a \in S_t} \tau_{t-1}[a] \implies x^* \in \arg \min_{a\in \mathcal{T}_{t-1}\setminus \{ z\pr\}} C_{t-1}\pr[a],
    %\]
    $x_t$ is added to $\mathcal{T}_t\pr$ with count $C_{t-1}\pr[z\pr]+1 = C_{t-1}[z]\geq C_{t-1}[x^*]$, and $x^*$ obtains the minimal count in $\mathcal{T}_t\pr$.
    The resulting state depends on on the difference between $C_{t-1}[x^*]$ and $C_{t-1}\pr[z\pr]$.
    First, if $C_{t-1}[x^*] = C_{t-1}\pr[z\pr]$, then
    \[
        C_t\pr[x^*] = C_{t-1}[x^*] = C_{t-1}\pr[z\pr] = C_{t-1}[z] -1  = C_t[z] -1.
    \]
    The counts of the items in the intersection remain unchanged and we remain in $\stwo$.
    Second, if $C_{t-1}[x^*] \neq C_{t-1}\pr[z\pr]$, 
    the condition on $C_{t-1}[z]$ implies that $C_{t-1}[z] =C_{t-1}[x^*] = C_{t-1}\pr[z\pr]+1$.
    Therefore, we observe 
    \[
        C_t[x_t] = C_{t-1}[x^*]+1 = C_{t-1}\pr[z\pr] +2 = C_t\pr[x_t] + 1 \quad (\text{condition}\sthree\text{(b)})  
    \]
    and 
    \[
        C_t[z] = C_{t-1}[z] =C_{t-1}[x^*] = C_t\pr[x^*] \quad (\text{condition}\sthree\text{(c)}).
    \]
    Therefore, the process transitions to $\sthree$
\end{itemize}
This concludes the case $x_t \notin \mathcal{T}_{t-1}\cup \mathcal{T}_{t-1}\pr$.

When $x_t \in W_{t-1}$, $C_{t-1}[x_t]$ is incremented  to form $C_t$ and $x_t$ is added to $\mathcal{T}_t\pr$.
If $x_t$ replaces the item $z\pr \in W_{t-1}\pr$ in $\mathcal{T}_t\pr$,
then $\mathcal{T}_t = \mathcal{T}_t\pr$, with $C[x_t] = C\pr[x_t]+1$, and we transition to $\sone$.
Else, if $x_t$ replaces some $x^* \notin W_{t-1}\pr$ then we observe $W_t = \{x^*\}$ and $W_t\pr = \{z\pr\}$.
Note that we cannot have $C_{t-1}[x_t] = C_{t-1}\pr[x^*]$ as we have both $C_{t-1}[x_t] = C_{t-1}[z\pr]+1$ and $C_{t-1}[z\pr]\geq C_{t-1}\pr[x^*]$.
Therefore, 
\[
    C_t[x_t] = C_{t-1}[x_t] + 1 = C_{t-1}\pr[x^*] + 2 = C_t\pr[x_t] +1\quad (\text{condition}\sthree\text{(b)})  
\]
and 
\[
    C_t[x^*] =C_{t-1}[x^*] = C_{t-1}\pr[x^*] = C_{t-1}\pr[z\pr] = C_{t}\pr[z\pr] \quad (\text{condition}\sthree\text{(c)}).
\]
Therefore, the process transitions to $\sthree$

When $x_t \in W_{t-1}\pr$, $C_{t-1}\pr[x_t]$ is incremented to form $C_t\pr$ and $x_t$ is added to $\mathcal{T}_t$. 
The resulting state depends on the rank of $z \in W_{t-1}$ in $\mathcal{T}_{t-1}$.
If $z = \arg \max_{a \in S_{t}} \tau_{t-1}[a]$, then $x_t = z\pr$ replaces $z$ in $\mathcal{T}_{t-1}$ to form $\mathcal{T}_t$ and we transition to $\sone$.
Otherwise, $x_t = z\pr$ replaces $x^*$ in $\mathcal{T}_{t-1}$ to form $\mathcal{T}_t$.
This leads to the change $x^* \in W_t\pr$.
With a now familiar argument, 
$x^*$ in the eviction register of $\mathcal{T}_{t-1}$ implies that $x^*$ has the minimum count in $\mathcal{T}_{t-1}\pr \setminus\{z\pr\}$. 
Thus, $x^*$ obtains the minimum value in $\mathcal{T}_t\pr$.
The resulting state depends on the relative values of $C_{t-1}[x^*]$ and $C_{t-1}[z]$.
If $C_{t-1}[x^*] = C_{t-1}[z] = C_{t-1}\pr[x_t]+1$, then 
\[
    C_t[z] = C_{t-1}[z] = C_{t-1}[x^*] =  C_t\pr[x^*].
\]
and 
\[
    C_t[x_t] = C_{t-1}[x^*]+1 = C_{t-1}\pr[x_t]+2 =  C_{t-1}\pr[x_t]+1,
\]
And we transition to $\sthree$.
Following the same line of logic, if $C_{t-1}[x^*] \neq C_{t-1}[z]$, we stay in $\stwo$. 

This concludes the possible state transitions from $\stwo$.
\end{proof}

\lemsthree*
\begin{proof}
    If the process is in $\sthree$, we have to consider the cases $x_t \in \mathcal{T}_{t-1} \cap \mathcal{T}_{t-1}\pr$, $x_t \notin \mathcal{T}_{t-1} \cup \mathcal{T}_{t-1}\pr$, $x_t \in W_{t-1}$, or $x_t\in W_{t-1}\pr$.
    Similar to prior cases, the occurrence $x_t \in \mathcal{T}_{t-1} \cap \mathcal{T}_{t-1}\pr$ does not change the state and we remain in $\sthree$.

    If $x_t \notin \mathcal{T} \cup \mathcal{T}\pr$, there are four possibilities, which are dependent on the identities of the items in the eviction registers.
\begin{itemize}
    \item[(i)] $\arg \max_{a\in S_t} \tau_{t-1}[a] = \arg \max_{a\in S_t\pr} \tau_{t-1}\pr[a]$. 
    This case follows the same logic as the `$x_t \in \mathcal{T}_{t-1}\cap \mathcal{T}_{t-1}\pr$' instance discussed above and we stay in $\sthree$.
    \item [(ii)] $z = \arg \max_{a\in S_t} \tau_{t-1}[a]$ and $z\pr = \arg \max_{a\in S_t\pr} \tau_{t-1}\pr[a]$.
    In this instance both $z$ and $z\pr$ are replaced by $x_t$ to form the sets $\mathcal{T}_t = \mathcal{T}_{t-1}$ (condition $\sone$(a)), with $C_t[x_t] = C_t\pr[x_t]$.
    All other counts are left unchanged and we transition to $\sone$.
    \item [(iii)] $z \neq x^* = \arg \max_{a\in S_t} \tau_{t-1}[a]$ and $z\pr = \arg \max_{a\in S_t\pr} \tau_{t-1}\pr[a]$ (the case $z = \arg \max_{a\in S_t} \tau_{t-1}[a]$ and $z\pr \neq \arg \max_{a\in S_t\pr} \tau_{t-1}\pr[a]$ is symmetrical).
    The following tracked item replacements occur 
    \begin{align*}
        \mathcal{T}_{t} &\gets (\mathcal{T}_{t-1}\setminus \{x^*\}) \cup \{x_t\}, \\
        \mathcal{T}_{t}\pr &\gets (\mathcal{T}_{t-1}\pr\setminus \{z\pr\}) \cup \{x_t\},
    \end{align*}
    leading to $W_t = \{z\}$ and $W_t\pr = \{x^*\}$.
    As $C_{t-1}[z] = C_{t-1}\pr[z\pr]$, we have $C_{t-1}[x^*] = C_{t-1}\pr[z\pr]$.
    Therefore,
    \[
        C_t[x_t] = C_{t-1}[x^*] + 1 = C_{t-1}\pr[z\pr] + 1 = C_t\pr[x_t]
    \]
    and
    \[
        C_t\pr[x^*] = C_{t-1}[x^*] = C_{t-1}[z] = C_t[z].
    \] 
    Note that
    \[
        C_t\pr[x^*] = \arg \min_{a \in \mathcal{T}_{t-1}} = \arg \min_{a \in \mathcal{T}_{t-1}\pr}   C_t\pr[a]\leq \arg \min_{a \in \mathcal{T}_t} C_t\pr[a].
    \]
    Therefore, with all other counts unchanged, the process stays in $\sthree$.
    \item [(iv)] $z \neq \arg \max_{a\in S_t} \tau_{t-1}[a] \neq \arg \max_{a\in S_t\pr} \tau_{t-1}\pr[a] \neq z\pr$.
    Here, we have items $y, y\pr \in \mathcal{T}\cap \mathcal{T}\pr$, with $y\neq y\pr$, in the eviction register. 
    By Lemma~\ref{lem:tau_implication}, this can only happen if $y\pr$ is the item where $C_{t-1}[y\pr] = C_{t-1}\pr[y\pr] + 1$.
    The following tracked item replacements occur 
    \begin{align*}
        \mathcal{T}_{t} &\gets (\mathcal{T}_{t-1}\setminus \{y\}) \cup \{x_t\}, \\
        \mathcal{T}_{t}\pr &\gets (\mathcal{T}_{t-1}\pr\setminus \{y\pr\}) \cup \{x_t\}.
    \end{align*}
    This leads to $W_t = \{z, y\pr\}$ and $W_t\pr = \{z\pr, y\}$. 
    As, $C_{t-1}[y]= C_{t-1}\pr[y\pr]$ (condition $\sthree$(c) at time $t-1$ states that $\mathcal{T}_{t-1}$ and $\mathcal{T}_{t-1}\pr$ have the same minimum count), we obtain $C_t[x_t]=C_t\pr[x_t]$.
    In addition, $\forall a \in \mathcal{T}_{t-1}\cap \mathcal{T}_{t-1}\setminus \{y\pr\}$ we have $C_{t-1}[a]=C_{t-1}\pr[a]$.
    Therefore, 
    \[
        \forall a \in \mathcal{T}_t \cap \mathcal{T}_t\pr = \mathcal{T}_{t-1}\cap \mathcal{T}_{t-1}\setminus \{y\pr\} \cup \{x_t\}: C_t[a] = C_t\pr[a], 
    \]
    Satisfying condition $\sfour$(b).
    In addition,
    \begin{align*}
        C_t[z] &= C_{t-1}[z] = C_{t-1}[y]= C_{t-1}\pr[y] = C_t\pr[y]  \\
        C_t\pr[z\pr] &= C_{t-1}\pr[z\pr] = C_{t-1}\pr[y\pr] = C_{t-1}[y\pr] - 1 = C_t[y\pr] -1.
    \end{align*}
    This is exactly condition $\sfour$(c).
    Lastly, as 
    \[
        C_t\pr[x_t] = C_{t-1}\pr[y\pr]+1 = C_{t-1}[y] +1 = C_t\pr[y] + 1,
    \]
    $x_t$ does not appear in the eviction register of $\mathcal{T}_t\pr$.
    In addition, after the arrival of $x_t$,
    \[
        y = \arg \max_{a \in S_t} \tau_{t-1}[a] \implies y = \arg \max_{a \in S_{t+1}\pr\setminus\{z\pr\}} \tau_t\pr[a].
    \]
    Therefore either $y$ or $z\pr$ is in the eviction register of $\mathcal{T}_t\pr$ (condition $\sfour$(d)), and a transition to $\sfour$ occurs.
\end{itemize}
This concludes the case $x_t \notin \mathcal{T}_{t-1}\cup \mathcal{T}_{t-1}\pr$.

When $x_t \in W_{t-1}$, $C_{t-1}[x_t]$ is incremented  to form $C_t$ and $x_t$ is added to $\mathcal{T}_t\pr$.
If $x_t$ replaces the item $z\pr \in W_{t-1}\pr$ in $\mathcal{T}_t\pr$,
then $\mathcal{T}_t = \mathcal{T}_t\pr$, with $C_t[x_t] = C_t\pr[x_t]$, and we transition to $\sone$.
Else, if $x_t$ replaces some $x^* \notin W_{t-1}\pr$, then we observe $W_t = \{x^*\}$ and $W_t\pr = \{z\pr\}$.
As $C_{t-1}[x_t] = C_{t-1}\pr[x^*]$,
we observe $ C_t[x_t] = C_t\pr[x_t]$.
%\begin{align*}
%    C_t[x_t] = C_t\pr[x_t] \quad &\text{if } C_{t-1}[x_t] = C_{t-1}\pr[x^*] \\
%    C_t[x_t] = C_t\pr[x_t] + 1 \quad &\text{if } C_{t-1}[x_t] = C_{t-1}\pr[x^*] +1 
%\end{align*}
Lastly, 
\begin{align*}
     C_t[x^*] =C_{t-1}[x^*] = C_{t-1}\pr[x^*] = C_{t-1}\pr[z\pr] &= C_{t}\pr[z\pr] \\
     &= \min_{a \in \mathcal{T}_{t-1}\pr} C_{t-1}\pr[a] \\
     &\leq  \min_{a \in \mathcal{T}_{t}\pr} C_{t}\pr[a].
\end{align*}
As the remaining counts are untouched,
the process remains in $\sthree$.

By the symmetry of $C_{t-1}[z] = C_{t-1}\pr[z\pr] = \min_{a \in \mathcal{T}_{t-1}\pr} C_{t-1}\pr[a]$, the case $x_t \in W_{t-1}\pr$ leads to the same state transitions as $x_t \in W_{t-1}$.

This concludes the possible state transitions from $\sthree$.

\end{proof}

\lemsfour*
\begin{proof}
    If the process is in $\sfour$, we have to consider the cases $x_t \in \mathcal{T}_{t-1} \cap \mathcal{T}_{t-1}\pr$, $x_t \notin \mathcal{T}_{t-1} \cup \mathcal{T}_{t-1}\pr$, $x_t \in W_{t-1}$, or $x_t\in W_{t-1}\pr$.
    Similar to prior cases, the occurrence $x_t \in \mathcal{T}_{t-1} \cap \mathcal{T}_{t-1}\pr$ does not change the state and we remain in $\sfour$.

    When $x_t \notin \mathcal{T}_{t-1} \cup \mathcal{T}_{t-1}\pr$, there are two possibilities to consider based on the identities of the items in the eviction registers.
    Note that $\arg \max_{a\in S_t} \tau_{t-1}[a] = \arg \max_{a\in S_t\pr} \tau_{t-1}\pr[a]$ is not possible as, by condition $\sfour$(d) at time $t-1$, the eviction register in $\mathcal{T}_{t-1}\pr$ is occupied by a member of $W_{t-1}\pr$.
    \begin{itemize}
        \item [(i)] $z_2 = \arg \max_{a\in S_t} \tau_{t-1}[a]$ and $z_2\pr = \arg \max_{a\in S_t\pr} \tau_{t-1}\pr[a]$.
        As $C_{t-1}[z_2] = C_{t-1}[z_2\pr]$, we observe the update $C_{t}[x_t]=C_t\pr[x_t]$ and $W_t=\{z_1\}$ and $W_t\pr = \{z_1\pr\}$.
        By condition $\sfour$(c) at time $t-1$, 
        \[
            C_t[z_1] = C_{t-1}[z_1] = C_{t-1}[z_2] + 1 = C_t[z_2]. 
        \]
        Therefore, with all other counters unchanged, we transition to $\stwo$.
        \item[(ii)] $z_2 \neq x^* =\arg \max_{a\in S_t} \tau_{t-1}[a]$ and $z_2\pr = \arg \max_{a\in S_t\pr} \tau_{t-1}\pr[a]$.
        The following tracked item replacements occur 
        \begin{align*}
            \mathcal{T}_{t} &\gets (\mathcal{T}_{t-1}\setminus \{x^*\}) \cup \{x_t\}, \\
            \mathcal{T}_{t}\pr &\gets (\mathcal{T}_{t-1}\pr\setminus \{z_2\pr\}) \cup \{x_t\},
        \end{align*}
        leading to $W_t = \{z_1, z_2\}$ and $W_t\pr = \{z_1\pr,x^*\}$.
        As both $\mathcal{T}_{t-1}$ and $\mathcal{T}_{t-1}\pr$ have the same minimum count, 
        \[
            C_t[x_t] = C_{t-1}[x^*] + 1 = C_{t-1}\pr[z_2\pr] + 1 = C_t\pr[x_t]
        \]
        and
        \[
            C_t\pr[x^*] = C_{t-1}[x^*] = \min_{a \in \mathcal{T}_{t-1}} C_{t-1}[a] = C_{t-1}[z_2] = C_t[z_2].
        \] 
        Note that the counts of $z_1, z_2$ and $z_1\pr$ are unchanged.  
        Further, by Lemma~\ref{lem:tau_implication},
        \[
            x^* =\arg \max_{a\in S_t} \tau_{t-1}[a] \implies x^* =\arg \max_{a\in S_{t+1}\pr\setminus\{z_2\pr\}} \tau_{t}\pr[a],
        \]
        which implies either $x^*$ or $z_2\pr$ occupy the eviction register of $\mathcal{T}_t\pr$.
        Therefore, the process stays in $\sfour$. 
    \end{itemize}
    This concludes the case $x_t \notin \mathcal{T}_{t-1}\cup \mathcal{T}_{t-1}\pr$.

    If $x_t \in W_t$, $C_{t-1}[x_t]$ is incremented  to form $C_t$ and $x_t$ is added to $\mathcal{T}_t\pr$.
    If $x_t$ replaces the item $z_1\pr \in W_{t-1}\pr$ in $\mathcal{T}_t\pr$ (the case is identical for $z_2\pr$), then the sets of isolated elements reduce in size:
    \[
        W_t = W_{t-1}\setminus \{x_t\}; \quad W_t\pr = W_{t-1}\pr \setminus \{ z_1\}.
    \]
    The resulting state depends on the identity of $x_t \in W_{t-1} = \{z_1, z_2$, with $C_{t-1}[z_1] = C_{t-1}[z_2]+1$.
    If $x_t = z_1$, we observe 
    \[ 
        C_t[x_t] = C_t[z_1] = C_{t-1} [z_1] + 1 = C_{t-1}[z_2] + 2 = C_{t-1}\pr[z_1\pr] +2 = C_t\pr[x_t] + 1.
    \]
    With all other counts unchanged, this results in a transition to $\sthree$.
    Otherwise, if $x_t = z_2$, we observe
    \[
        C_t[x_t] = C_t[z_2] = C_{t-1} [z_2] + 1 = C_{t-1}\pr[z_1\pr] + 1  = C_t\pr[x_t].
    \]
    As $C_t[z_1] = C_t\pr[z_2\pr]$, the process transitions to $\stwo$.
    
    If $x_t = z_1\pr\in W_{t-1}\pr$ (the case $x_t = z_2\pr$ is identical), $C_{t-1}\pr[x_t]$ is incremented  to form $C_t\pr$ and $x_t$ is added to $\mathcal{T}_{t-1}$ to form $\mathcal{T}_t$.
    If $x_t$ replaces $z_2 \in W_{t-1}$ in $\mathcal{T}_{t-1}$, we observe $W_t = {z_1}$ and $W_{t}\pr = \{ z_2\pr \}$, with $C_t[z_1] = C_t\pr[z_2\pr] +1$ and $C_t[x_t] = C_t\pr[x_2]$, and the process transitions to $\stwo$.
    Otherwise, if $x_t$ replaces some $x^* \notin W_{t-1}$, we observe $W_t = {z_1, z_2}$ and $W_{t}\pr = \{ z_2\pr, x^* \}$ with
    \[
        C_t[z_1] = C_t[z_2] + 1 = C_{t-1}[x^*] + 1 = C_t\pr[x^*] +1 = C_t\pr[z_2\pr] + 1 . 
    \]
    As $x^*$ was in the eviction register of $\mathcal{T}_{t-1}$, after the arrival of $x_t$, either $x^*$ or $Z_2\pr$ are in the eviction register of $\mathcal{T}_t\pr$.
    Thus, we remain in $\sfour$.

    This concludes the possible state transitions from $\sfour$.
\end{proof}

\subsection{Performance and Utility}
\label{app:sub_utility}

\thmutility*
\begin{proof}
    Set $\tilde{k}=2k$ and run $\spacesaving$ with capacity $\tilde{k}$.
    This uses $\mathcal{O}(\tilde{k})=\mathcal{O}(k)$ words and requires $\mathcal{O}(1)$ update time using a standard implementation \cite{cormode2008finding}.
    Let $C$ be the terminal (non-private) counters and add independent noise $\eta_x\sim\laplacedist(1/\varepsilon)$, for each item $x\in \mathcal{T}$.
    Release labels with $C[x]+\eta_x>\tau$. 
    As
    \[
    \frac{|X|}{k} - \gamma = 2\cdot \frac{|X|}{2k} - \gamma > \frac{|X|}{2k} + \gamma +1,
    \]
    where the inequality follows by assumption, it follows that $\tau = \tfrac{|X|}{k}-\gamma$.
    By the error guarantee of $\spacesaving$ (Lemma~\ref{lem:ss_bound}),
    \[
    f_x \;\le\; C[x] \;\le\; f_x + \frac{|X|}{\tilde{k}} = f_x + \frac{|X|}{2k},
    \]
    for all $x \in \mathcal{T}$.
    For $\eta_x\sim \mathrm{Laplace}(1/\varepsilon)$ we have
    $\Pr[|\eta_x|>t]=e^{-\varepsilon t}$.
    Taking $t=\tfrac{1}{\varepsilon}\ln\left(\tfrac{1}{\delta}\right)$ gives, with probability $\ge 1-\delta$,
    \[
    f_x - \frac{1}{\varepsilon}\ln\!\Big(\frac{1}{\delta}\Big)
    \;\le\; C[x]+\eta_x
    \;\le\; f_x + \frac{|X|}{2k} + \frac{1}{\varepsilon}\ln\!\Big(\frac{1}{\delta}\Big),
    \]
    which matches the stated bound.

    If $f_y>|X|/k$, then $C[y]\ge f_y > |X|/k$, so, with probability $1-\delta$,
    \[
    C[y]-\tau \;>\; \frac{|X|}{k}-\Big(\frac{|X|}{k}-\gamma\Big)
    =\gamma \;>\;0
    \]
    by the premise $|X|/(2k)>\gamma$.
    Therefore, $y$ fails the threshold only if $\eta_y<-\gamma$, which occurs with probability $\le \tfrac{1}{2}e^{-\varepsilon\gamma}< \delta$.
\end{proof}

\section{Deferred Proofs Section~\ref{sec:lin_sketch}}
\label{app:proof_2}
\privfreqoracle*

\begin{proof}
    For any fixed set of input parameters $(k, \tilde{k})$ define
    \[
        \mathcal{A}(X) = \mathsf{EEHH}(X, k, \tilde{k}, \varepsilon, \delta).
    \]
    Let $\mathcal{T}$ and $\mathcal{T}\pr$ be the outputs of Algorithm~\ref{alg:topk} on, respectively, inputs $X$ and $X\pr$.
    The outputs from Algorithm~\ref{alg:EEHH} are defined analogously. 
    We prove that for any pair of neighboring inputs $X,X'$ (where $X'$ is obtained by removing a single element from $X$), and any measurable subset $S$ of outputs,
    \[
    \Pr[\,\mathcal{A}(X)\in S\,] \;\le\; e^{\varepsilon}\,\Pr[\,\mathcal{A}(X')\in S\,] \;+\; \delta.
    \]
    
    Define the ``stable envelope'' event $\xi$ as the event that no frequency query exceeds the error envelope during the execution of the algorithm.
    That is, for
    \begin{align*}
        E_1 \;&:=\;
    \bigl[\,
        \gamma_1(t)\ \ge\ \tilde{f}_t(x_t)-f_t(x_t)\ \ge\ \gamma_2(t) , \forall x_t \in X    
    \,\bigr] \\
        E_1\pr \;&:=\;
    \bigl[\,
        \gamma_1(t)\ \ge\ \tilde{f}_t(x_t\pr)-f_t(x_t\pr)\ \ge\ \gamma_2(t),  \forall x_t\pr \in X\pr    
    \,\bigr] \\
    E_2 \;&:=\;
    \bigl[\,
        \gamma_1(T)\ \ge\ \tilde{f}_T(x)-f_T(x)\ \ge\ \gamma_2(T),  \forall x \in \mathcal{T}    
    \,\bigr] \\
    E_2\pr \;&:=\;
    \bigl[\,
        \gamma_1(T)\ \ge\ \tilde{f}_T(x\pr)-f_T(x\pr)\ \ge\ \gamma_2(T),  \forall x\pr \in \mathcal{T}\pr    
    \,\bigr]
    \end{align*}
    $\xi = E_1\land E_1\pr \land E_2\land E_2\pr$.
    There are at most $T+\tilde{k}$ frequency queries during algorithm execution.
    Thus, using a union bound across both streams, by assumption, $\Pr[\xi]\geq 1-\delta$.

    The threshold is set at $\tau > T/\tilde{k} + 2\gamma_1(T) + \gamma_2(T)$.
    Each estimate during post-processing has error at most $\gamma_1(T)$.
    By Lemma~\ref{lem:suppress}, an item in the symmetric difference has error at most
    \[
      \hat{f}_T(x) \leq f_T(x) + \gamma_1(T) \leq \frac{t}{\tilde{k}} + \gamma_2(T) + 2\gamma_1(T) \leq \tau,  
    \]
    where $t\leq T$ is the last time it was removed from the tracked set in the opposing stream. 
    Therefore, conditioned on event $\xi$, the released index sets $\hat{\mathcal{T}}$ and $\hat{\mathcal{T}}\pr$ are contained in the intersection $\mathcal{T}\cap\mathcal{T}\pr$.
    Let $m := |\mathcal{T}\cap \mathcal{T}\pr|\ge \tilde{k}-2$, and fix an arbitrary but common ordering of the labels in $\mathcal{T}\cap \mathcal{T}\pr$.
    Define the vector-valued function
    \(
    F(X) \in \mathbb{R}^m
    \)
    as the ordered list of frequency estimates.
    \(
    F(X) := (\, \mathcal{D}.\mathsf{query}(x) \,)_{x\in\mathcal{T}\cap\mathcal{T}\pr},
    \)
    and similarly $F(X\pr) := (\, \mathcal{D}\pr.\mathsf{query}(x) \,)_{x\in\mathcal{T}\cap\mathcal{T}'}$.
    By assumption, for all measurable $A\subseteq \mathbb{R}^m$
    \begin{align}
        \Pr[F(X) \in A] &\leq e^{\varepsilon} \Pr [F(X\pr) \in A].
        \label{eqn:dp-assumption}
    \end{align}
    
    Now define the (deterministic) post-processing map $\Phi$
    that keeps precisely those coordinates of $F(X)$ whose value exceeds $\tau$, together with their labels.
    Conditioned on $\xi$, 
    by the post-processing property of differential privacy, applying $\Phi$ preserves Inequality~\eqref{eq:laplace-mech}.
    That is, for all measurable $S,$
    \begin{equation}\label{eq:cond-dp1}
        \Pr[\, \Phi(F(X) ) \in S \mid \xi \,]
        \;\le\;
        e^{\varepsilon}\,
        \Pr[\, \Phi(F(X\pr)) \in S \mid \xi \,].
    \end{equation}

    Putting everything together,
    using~\eqref{eq:cond-dp1}, $\Pr[E]\le 1$, and $\Pr[\neg E]\le \delta$, we obtain
    \begin{align*}
    \Pr[\,\mathcal{A}(X)\in S\,] 
    &= 
    \Pr[\,\mathcal{A}(X)\in S\, \mid \xi] \Pr[\xi] \\
    &+
    \Pr[\,\mathcal{A}(X)\in S\, \mid \neg \xi] \Pr[\neg \xi] \\
    &\le
    e^{\varepsilon}\,\Pr[\,\mathcal{A}(X')\in S \mid \xi\,]\,\Pr[\xi] \\
    &+
    \Pr[\,\mathcal{A}(X)\in S \mid \neg \xi\,]\,\Pr[\neg \xi]
    \\
    &\le
    e^{\varepsilon}\,\Pr[\,\mathcal{A}(X')\in S\, \mid \xi] \Pr[\xi]
    \;+\;
    \Pr[\neg E]
    \\
    &\le
    e^{\varepsilon}\,\Pr[\,\mathcal{A}(X')\in S]
    \;+\;
    \Pr[\neg E]
    \\
    &\le
    e^{\varepsilon}\,\Pr[\,\mathcal{A}(X')\in S\,]
    \;+\;
    \delta,
\end{align*}
as required.
\end{proof}

\end{document}